\documentclass[10pt,a4paper,openright]{article}
\usepackage[hscale=0.7,vscale=0.8]{geometry}

\usepackage{booktabs} 

\usepackage[ruled]{algorithm2e} 

\SetAlFnt{\small}
\SetAlCapFnt{\small}
\SetAlCapNameFnt{\small}
\SetAlCapHSkip{0pt}
\IncMargin{-\parindent}

\usepackage{graphicx}
\usepackage[cmex10]{amsmath}
\usepackage{enumerate}
\usepackage{amssymb}
\usepackage{amsthm}

\DeclareMathOperator*{\pr}{\mathbb{P}}

\newtheorem{theorem}{Theorem}
\newtheorem{corollary}{Corollary}
\newtheorem{lemma}{Lemma}
\newtheorem{proposition}{Proposition}

\usepackage{url}

\usepackage{subfigure} 
\usepackage{multirow}
\usepackage{datetime}
\usepackage{longtable}
\usepackage{epstopdf}

\usepackage{xr}
\usepackage[final]{changes}
\definechangesauthor[color=red]{AP}

\definechangesauthor[color=blue]{CB}

\newcommand{\addCB}[1]{\added[id=CB]{#1}}

\begin{document}
\title{What you lose when you snooze: how duty cycling impacts on the contact process in opportunistic networks}

\author{Elisabetta Biondi, Chiara Boldrini, \\Andrea Passarella, Marco Conti\footnote{All authors are with IIT-CNR, Via G. Moruzzi 1, Pisa, Italy}}

\date{}

\maketitle

\begin{abstract}
In opportunistic networks, putting devices in energy saving mode is crucial to preserve their battery, and hence to increase the lifetime of the network and foster user participation. A popular strategy for energy saving is duty cycling. However, when in energy saving mode, users cannot communicate with each other. The side effects of duty cycling are twofold. On the one hand, duty cycling may reduce the number of usable contacts for delivering messages, increasing \emph{intercontact} times and delays. On the other hand, duty cycling may break long contacts into smaller contacts, thus also reducing the capacity of the opportunistic network. Despite the potential serious effects, the role played by duty cycling in opportunistic networks has been often neglected in the literature. In order to fill this gap, in this paper we propose a general model for deriving the pairwise contact and  intercontact times measured when a duty cycling policy is superimposed on the original encounter process determined only by node mobility. The model we propose is general, i.e., not bound to a specific distribution of contact and intercontact times, and very accurate, as we show exploiting two traces of real human mobility for validation. Using this model, we derive several interesting results about the properties of measured contact and intercontact times with duty cycling: their distribution, how their coefficient of variation changes depending on the duty cycle value, how the duty cycling affects the capacity and delay of an opportunistic network. The applicability of these results is broad, ranging from performance models for opportunistic networks that factor in the duty cycling effect, to the optimisation of the duty cycle to meet a certain target performance.
\end{abstract}

\section{Introduction}
\label{sec:intro}
 
The widespread availability of smart, handheld devices like smartphones and tablets has stimulated the discussion and research about the possibility of new concepts for supporting communications between users. Particularly appealing, towards this direction, is the opportunistic networking paradigm, in which messages arrive to their final destination through consecutive pairwise exchanges between users that are in radio range of each other~\cite{conti2015manet}. As such, unlike MANETs, opportunistic networks do not assume a continuous end-to-end path between source and destination, and paths are built dynamically and incrementally by intermediate nodes when new contacts (i.e., new forwarding opportunities) arise. While originally studied as a standalone solution, opportunistic networks are now being exploited in synergy with the cellular infrastructure in mobile data offloading scenarios~\cite{rebecchi2015data}, as an enabling technology for the Internet of Things~\cite{wirtz2014opportunistic}, and they are being enhanced to also exploit a cloud infrastructure when available~\cite{wirtz2013interest-based}. 

User mobility, and especially user encounters, is the key enabler of opportunistic communications. Unfortunately, ad hoc communications tend to be very energy hungry \cite{friedman2013power} and no user will be willing to participate in an opportunistic network if they risk to see their battery drained in a few hours. However, there are very few contributions that study how power saving mechanisms impact on the contacts that can be exploited to relay messages. These power saving mechanisms range from completely turning off devices periodically or, more commonly, to tuning the frequency at which the network interface is used (e.g., reducing neighbour discovery activities). We generically refer to all these strategies as \emph{duty cycling}. With duty cycling, messages can be exchanged only when two nodes are in one-hop radio range \emph{and} they are both in the active state of the duty cycle. So, power saving may reduce forwarding opportunities, because contacts are missed when at least one of the devices is in a low-energy state. Since some contacts may be missed, the \emph{measured intercontact times}, defined as the time interval between two consecutive detected encounters between the same pair of nodes, is, in general, larger than the original intercontacts (i.e., those defined exclusively by the nodes mobility process) and this may clearly affect the delay experienced by messages. The \emph{measured contacts} (i.e., the length of a contact while the two nodes are in radio range and active) may also be affected, if one of the two nodes becomes inactive during a contact.
Owing to the extent at which, in principle, measured contacts and intercontact times may affect the performance of opportunistic networks, we argue that it is essential to better understand how they are characterised and how they depend on the duty cycling policy in use. Unfortunately, the effects of duty cycling on the measured pairwise contact process have been largely ignored in the literature.

The goal of this work is to characterise the distribution of the \emph{measured} contact and intercontact times, starting from a given distribution of \emph{original} contact and intercontact times, and a duty cycling scheme. To this aim, the contribution of this paper is threefold. First, in Section~\ref{sec:ict_negligible}, assuming that contact duration is negligible, we derive a mathematical model of the measured intercontact times between nodes. 
For general distributions of the original intercontact times we derive mathematical expressions that can be solved numerically to obtain the first two moments of the measured intercontact times.  We can thus approximate \emph{any} distribution of the measured intercontact times using hyper- or hypo- exponential approximations~\cite{tijms2003first}. Under the two most popular intercontact time distributions considered in the literature, exponential and Pareto, the closed forms of the first two moments admit analytic solutions, making the model even more flexible. 
%
As a second contribution, in Section~\ref{sec:ict_non_negligible} we extend the above model to include the effects of non-negligible contact duration, again under any distribution of contact and intercontact times, thus making the model as general as possible. We extensively validate this model using as input the distribution of contact and intercontact times obtained from traces of real user mobility. Finally, in Section~\ref{sec:exp_dc} we show that the results obtained assuming a deterministic duty cycling (as described in Section~\ref{sec:contact_process}), actually provide a very good characterisation also when stochastic duty cycling is used.

\noindent Focusing on a tagged node pair, the key findings presented in this paper are the following:
\begin{itemize}
\item The measured contact time $\tilde{C}$ cannot be longer than the duration of the active state of the duty cycle, hence the data transfer capacity of the opportunistic network is generally reduced, even significantly. However, if the duty cycling policy is such that nodes refrain from entering the low-power state when they detect a contact, the measured contact duration (hence the capacity) may be only minimally affected by the duty cycle. \\
\item When contact duration $C$ is negligible, if the original intercontact times $S$ are exponential with rate $\lambda$, the measured intercontact times $\tilde{S}$ are exponential with rate $\lambda \Delta$, where $\Delta$ is the percentage of time nodes keep the wireless interface active (duty cycling parameter). Instead, if the original intercontact times $S$ are Pareto with exponent $\alpha$, the measured intercontact times $\tilde{S}$ do not feature a well-known distribution but they decay as a Pareto random variable with the same exponent $\alpha$. This implies that all the properties (e.g., the delay convergence~\cite{chaintreau2007impact,boldrini2015stability}) that depend on the shape of the tail of the Pareto distribution of intercontact times are not affected by duty cycling. \\
\item The duty cycle can affect measured intercontact times in such a way that low-variability (i.e., with coefficient of variation smaller than one) original intercontact times can turn into highly variable measured intercontact times (and vice versa), thus potentially altering the convergence of the expected delay. \\
\item A stochastic duty cycling can be approximated with a deterministic duty cycling for which the length of the active and inactive intervals corresponds to the average length of the same intervals in the stochastic duty cycling. This means that our results about $\tilde{C}$ and $\tilde{S}$ hold for a very large class of duty cycling policies.
\end{itemize}

\noindent To the best of our knowledge, as discussed in Section~\ref{sec:relwork}, this work represents the first comprehensive analysis of how the measured contact process is altered by power saving techniques.

\section{Preliminaries}
\label{sec:preliminaries}

In this section we introduce the duty cycling process that we take as reference and we describe how the contacts between users can be modelled.

\subsection{The duty cycling process}
\label{sec:dc_process}

We use duty cycling in a general sense, meaning any power saving mechanism that hinders the possibility of a continuous scan of the devices in the neighbourhood. We assume that nodes alternate between the ON and OFF states. In the ON state, nodes are able to detect contacts with other devices. In the OFF state (which may correspond to a low-power state or simply to a state in which devices are switched off) contacts with other devices are missed. Using this generalisation, we are able to abstract from the specific wireless technology used for pairwise communications.

Duty cycling policies can be deterministic or stochastic, depending on how the length of their ON and OFF states is chosen (fixed, in the former case, varying according to some known probability distribution in the latter). In the literature also non-stationary duty cycling policies can be found, in which the length of ON and OFF states depends on some properties of the network (or a node's neighborhood) at time $t$. All these approaches are discussed in Section~\ref{sec:relwork}. 
We base our model on the deterministic duty cycling case (which requires a coarse synchronisation between devices), then we later prove in Section~\ref{sec:exp_dc} how this model captures the average behaviour of the stochastic case (which does not require synchronisation) as well. Finally, we also discuss how the model captures some notable cases of non-stationary duty cycling policies.

In the following, we assume that the duty cycle process and the real contact process are independent and, considering a tagged node pair, we denote with $\tau$ the length of the time interval in which both nodes are ON, and with $T$ the period of the duty cycle. Thus, $T-\tau$ corresponds to the duration of the OFF interval and $\Delta=\frac{\tau}{T}$ is the duty cycle parameter.
%
In general, the ON interval can start anywhere within $T$ but here, without loss of generality, we assume it starts at the beginning of interval $T$. In addition, considering a generic detected contact, we count duty cycle periods from the first one where the contact is detected.
Hence, ON intervals will be of type $[i T, \tau + i T)$, with $i\geq 0$ and OFF intervals of type $[\tau + iT, (i+1)T)$, with $i \geq 0$ (Figure~\ref{fig:on_off_intervals}). In the following, ON and OFF intervals will be denoted with $\mathcal{I}^{ON}$ and $\mathcal{I}^{OFF}$, respectively. Hence, the set of all ON (OFF) intervals is given by $\bigcup_{n=0}^{\infty}\mathcal{I}^{ON}_n$ ($\bigcup_{n=0}^{\infty}\mathcal{I}^{OFF}_n$).
Focusing on a tagged node pair, we can represent how the duty cycle function evolves with time as  $d(t) = \left\{ \begin{array}{ll}
1& \textrm{if } t \bmod T \in [0, \tau)\\
0 &  \textrm{ otherwise}
\end{array} \right.$. 
When $d(t)=1$, both nodes are ON, thus their contacts, if any, are detected. 

\begin{figure}[t]
\begin{center}
\includegraphics[scale=0.5, angle=90]{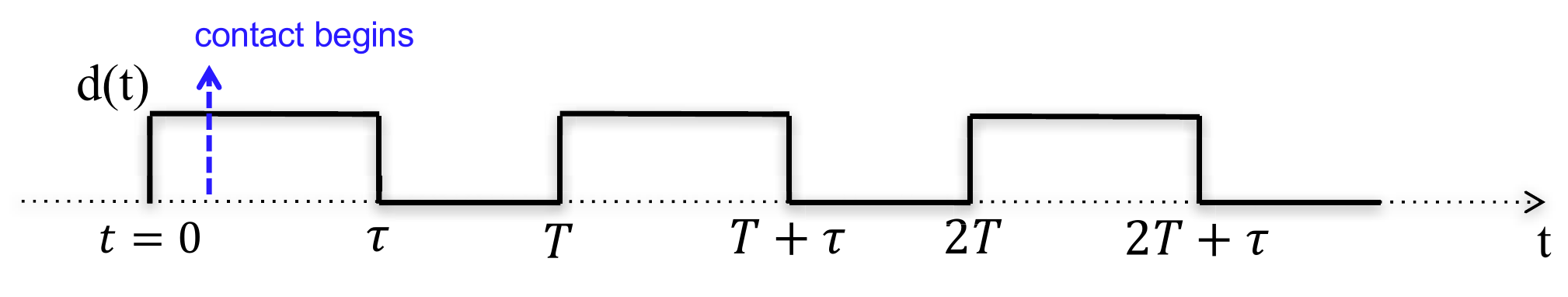}
\caption{ON and OFF intervals in our model.} \vspace{-20pt}
\label{fig:on_off_intervals}
\end{center}
\end{figure}

Reference values for $\tau$ and $T$ depend also on practical aspects. A result derived in~\cite{trifunovic2014adaptive} shows that, for frequencies of switching between ON and OFF states beyond $1/100 s^{-1}$, energy consumption on smartphones drastically increases, making the discovery process not energy efficient. Thus, for the purpose of this paper, we will consider values of $T$ around $100$s and will experiment with different values of $\tau$ when evaluating the proposed models. This is also the reason why in the paper we do not consider duty cycling schemes switching the wireless interfaces at a much finer granularity, in the order of milliseconds or less.

\subsection{The contact process}
\label{sec:contact_process}

Similarly to the related literature \cite{picu2012analysis,boldrini2014performance}, we assume that, from the mobility standpoint, node pairs are independent. From the modelling standpoint, the contact process of each node pair $(u,w)$ can be approximated as an \emph{alternating renewal} process~\cite{cox1962renewal}. In this case, the node pair alternates between the CONTACT state in which the two nodes are in radio range, and a state in which they are not (Figure~\ref{fig:alternate_renewal}). The time interval between the beginning and the end of the $i$-th contact is denoted as $C_i^{(u,w)}$. The time interval between the end of a contact and the beginning of the next one corresponds to the intercontact time and it is denoted as $S_i^{(u,w)}$. Hence, the alternating renewal process corresponds to the independent sequence of random variables $\{ C_i^{(u,w)}, S_i^{(u,w)}\}$, with $i \ge 1$, which is an approximation of the real contact process as $C_i^{(u,w)}$ and $S_i^{(u,w)}$ can be dependent for a fixed $i$ but must be independent for different $i$. Note that assuming independence of consecutive contact and intercontact times is also customary in the literature. Since in the following we focus on a tagged node pair,  for the sake of clarity we hereafter drop superscript~$(u,w)$ from our notation. Please note, however, that the contact process we consider is \emph{heterogeneous}, i.e., the distribution of $C_i$ and $S_i$ can be different for different pairs of nodes. Exploiting this notation, we have that the time $X_i$ at which the $i$-th contact begins is $\sum_{j=1}^{i} S_j + \sum_{j=0}^{i-1} C_j, \forall i \geq 1$, while the time $Y_i$ at which it ends is given by $\sum_{j=1}^{i} S_j + \sum_{j=0}^{i} C_j, \forall i \geq 1$.

\begin{figure}[h]
\begin{center}\vspace{-10pt}
\includegraphics[scale=0.5, angle=90]{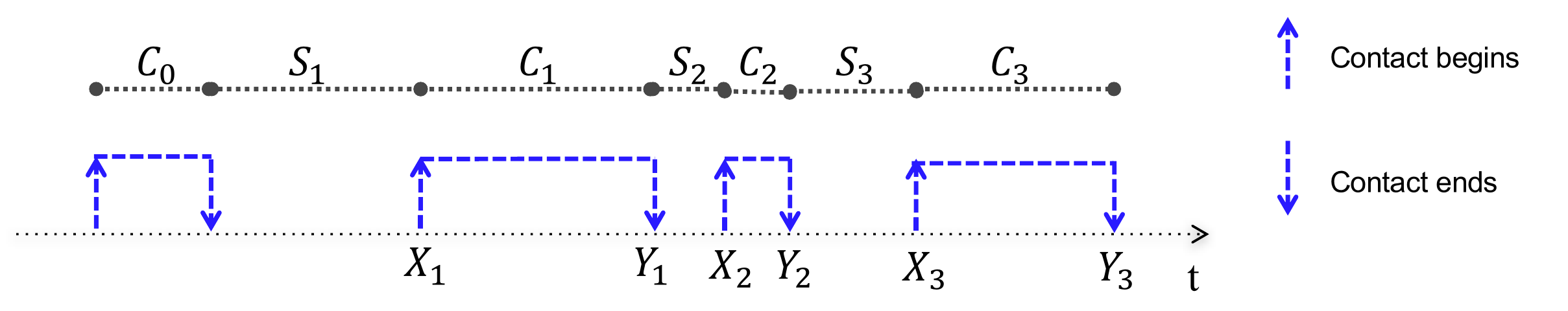}
\caption{The alternate-renewal contact process}\vspace{-10pt}
\label{fig:alternate_renewal}
\end{center}
\end{figure}

\section{Measured intercontact times when contact duration is negligible}
\label{sec:ict_negligible}

We now discuss how the \emph{measured} contact process depends on the contact process described in Section~\ref{sec:contact_process}. We start analysing the case where the contact durations are negligible with respect to the duty cycling period. In real traces, this is often a reasonable approximation, given that a significant part of observed distributions is concentrated on contact durations of a few seconds. Under this assumption, the alternating renewal process of Figure~\ref{fig:alternate_renewal} becomes a simple renewal process~\cite{cox1962renewal} where $S_i$ are the renewal intervals. We later relax this assumption in Section~\ref{sec:ict_non_negligible}.

Recall that the effect of duty cycling on contacts is that some contacts between nodes may be lost. So, we first study in Section~\ref{sec:ict_negligible_preliminaries} what are the characteristics of the process of measured contacts and how it can be modelled. The main outcome of this section is that the measured intercontact time can be obtained as the sum of $N$ (with $N$ stochastic) real intercontact times, assuming that the PDF of $S$ has certain properties that can simplify our derivations. Building upon these results, in Section~\ref{sec:ict_negligible_pdf_n} we compute the PMF of $N$ and then in Section~\ref{sec:ict_negligible_detected_ict} we finally derive the measured intercontact times~$\tilde{S}_i$.

\subsection{Problem setting}
\label{sec:ict_negligible_preliminaries}

Let us denote with $\tilde{S}_j$ the time between the $(j-1)$-th and the $j$-th detected contact (corresponding to the $j$-th \emph{measured} intercontact time) and assume that at time $\tilde{X}_{j-1}$ a contact has been detected, as shown in Figure~\ref{fig:contacts}. For convenience of notation, in Figure~\ref{fig:contacts} and in the following, the sequence number of the duty cycling interval in which the $i$-th contact takes place is denoted with $n_i$, while the sequence number of the ON interval in which the $j$-th \emph{detected} contact takes places is denoted with $\tilde{n}_{j}$.
\begin{figure}[h]
\begin{center}\vspace{-10pt}
\includegraphics[scale=0.5,angle=90]{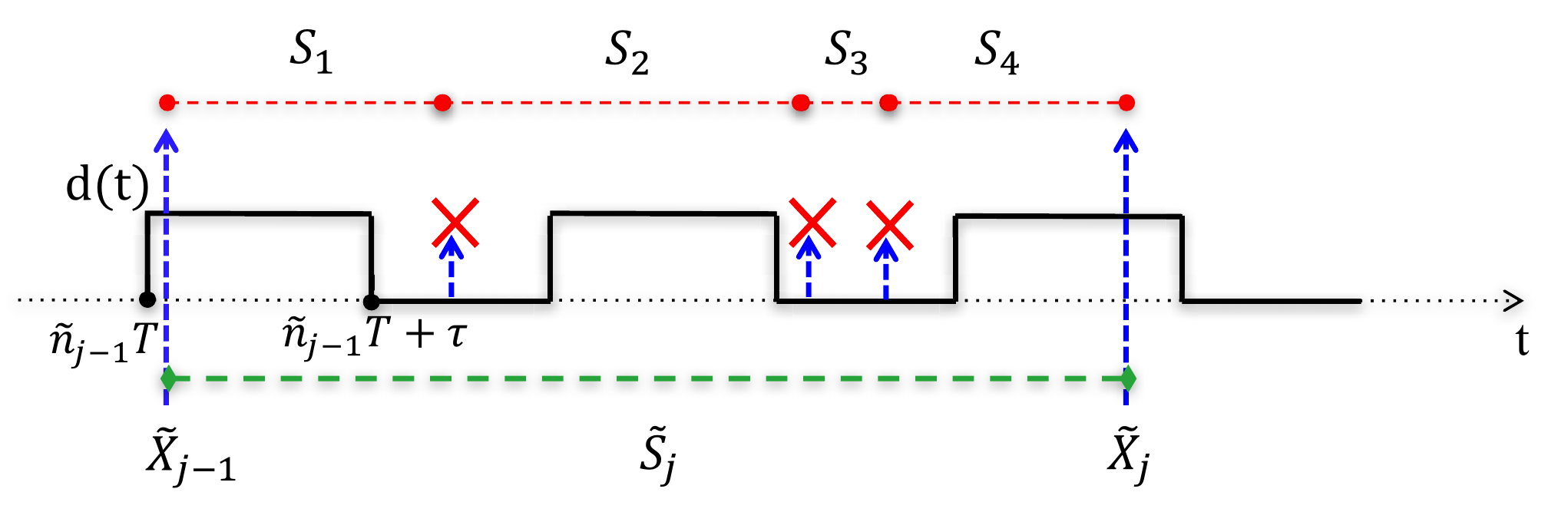}
\caption{Contact process with duty cycling}\vspace{-10pt}
\label{fig:contacts}
\end{center}
\end{figure} 

Then, the measured intercontact time $\tilde{S}_j$ is the time from $\tilde{X}_{j-1}$ until the next detected contact at $\tilde{X}_{j}$, which can be obtained by adding up the intercontact times $S_i$ between $\tilde{X}_{j-1}$ and the next detected contact. Denoting with $N_j$ the random variable measuring the number of intercontact times between the $(j-1)$-th and the $j$-th detected contact (e.g. in Figure~\ref{fig:contacts}, $N_{j} = 4$), we obtain $\tilde{S}_j = \sum_{i=1}^{N_j} S_i$, hence $\pr(\tilde{S}_j = x) = \sum_{k=1}^{\infty} \pr(\sum_{i=1}^{k} S_i = x)\pr(N_j=k)$. In general, $\tilde S_j$ are not i.i.d., because $N_j$ are in general not i.i.d. However, to keep the analysis tractable, we derive a condition under which $N_j$ can be assumed i.i.d., and thus $\tilde S_j$ also become i.i.d. This makes also the \emph{detected} intercontact time process a renewal one. Intuitively (see~\cite{biondi:what_you_lose:tr} for a formal proof\footnote{Content in~\cite{biondi:what_you_lose:tr} not included in this paper is provided as supplemental material.}) $N_j$ are i.i.d. when the probability density of~$S$ does not vary much inside an ON or OFF interval, and thus the exact time within an interval when a contact happens can be approximated as uniformly distributed in that interval, regardless of when the previous contacts took place. When this happens, in fact, since intercontact times $S_i$ are assumed i.i.d., the process forgets about its past and regenerates itself. 
\begin{lemma}[Uniformity in ON intervals]\label{lemma:uniformity}
When $f_{S_i}$ varies slowly\footnote{In the context of this paper, a function $f(x)$ varies slowly in a given interval if, for any $x_1,x_2$ belonging to that interval, $\frac{f_S(x_1)}{f_S(x_2)} \sim 1$. We are not implying that the function is slowly varying in the sense of~\cite{bingham1989regular}.} in intervals of length $\tau$, $\{N_j \}_{j \ge 1}$ can be modelled as i.i.d. (hence, $N_j \sim N$), and the displacement $Z^{ON}$ of a detected contact within an ON interval is approximately distributed as $Unif(0,\tau)$.
\end{lemma}

Lemma~\ref{lemma:uniformity} above applies to detected contacts. For modelling purposes, it is also convenient to approximate the displacement of missed contacts as uniformly distributed in the OFF interval in which they are missed. For this reason, we introduce the following lemma, which can be proved using similar arguments as those used in the proof of Lemma~\ref{lemma:uniformity}.

\begin{lemma}[Uniformity in OFF intervals]\label{lemma:uniformity_off}
When $f_{S_i}$ varies slowly in intervals of length $T-\tau$, it holds that the displacement of a missed contact within its OFF interval is distributed as $Z^{OFF} \sim Unif(0,T-\tau)$.
\end{lemma}

When Lemma~\ref{lemma:uniformity} holds true, the measured contact process is, at least approximately, a renewal process.
Thus, we can express $\tilde{S}$ as a random sum of i.i.d. random variables, , i.e. $\tilde S = \sum S_i$. Random sums have some useful properties that we will exploit in Section~\ref{sec:ict_negligible_detected_ict} in order to derive the first two moments of $\tilde{S}$. Please note that this formula is general, i.e., holds for any type of continuous intercontact time distribution and for any type of duty cycling policy. To complete the analysis of $\tilde S$, in Section~\ref{sec:ict_negligible_pdf_n} we compute the distribution of $N$, and in Section~\ref{sec:ict_negligible_detected_ict} we find the moments of $\tilde S$.

Before continuing, note that, as described in more detail in~\cite{biondi:what_you_lose:tr}, \addCB{we can derive simple conditions under which Lemma~\ref{lemma:uniformity} and Lemma~\ref{lemma:uniformity_off} hold, for two popular cases of intercontact times distributions, i.e. exponential and Pareto. Specifically, sufficient conditions are $\lambda T \ll1$ where $\lambda$ is the rate of exponential intercontact times, and $T \ll b$ where $b$ is the scale of Pareto intercontact times. Using these two conditions $\lambda T \ll 1$ and $T \ll b$, in Section~\ref{sec:ict_negligible_stilde_validation}, we will see how close theoretical predictions for measured intercontact times are to simulation results depending on whether these conditions are satisfied or not.}




\subsection{Deriving the distribution of $N$}
\label{sec:ict_negligible_pdf_n}

In this section, we derive the probability distribution of~$N$, defined as the number of contacts needed, after a detected contact, in order to detect the next one.
Since we are assuming the detected contact process to be renewal, we can focus on the portion of this process between two detected contacts. Specifically, we can focus, without loss of generality, on what happens between the first and second detected contact. Then, the rationale behind the derivation of $N$ is pretty intuitive. In fact, $N=1$ corresponds to the case where the first intercontact time after a detection ends in an ON interval. For case $N=2$, the first intercontact ends in an OFF interval, while the following one ends in a following ON interval. 
All other cases follow using the same line of reasoning. Please recall that, in the following, ON and OFF intervals will be denoted with $\mathcal{I}^{ON}$ and $\mathcal{I}^{OFF}$, respectively. 

The derivation of the PMF of $N$ in Theorem~\ref{theo:pdf_n_approx} below quantifies the probability $P\{N=k\}$. The line of reasoning for deriving this result is as follows. Let us define random variable $E_k$, which is equal to one when the $k$-th contact is in an ON interval, equal to zero otherwise. It is easy to see that the following holds true:
\begin{equation}\label{eq:n_conditioned_with_events}
\mathbb{P}\left(N=k  \right) \!=\! \mathbb{P}\left(E_1 =0, ..., E_{k-1}=0, E_k = 1\right).
\end{equation}
%

Similarly to the argument developed in Section~\ref{sec:ict_negligible_preliminaries}, this is because the probability that an intercontact falls in an ON or OFF interval depends on the point in time when the previous contact finishes. In~\cite{biondi2014duty} \addCB{we considered these dependencies and derived closed form solutions for exponential intercontact times. Hereafter we provide a solution for general intercontact time distributions under the conditions of Lemma~\ref{lemma:uniformity} and Lemma~\ref{lemma:uniformity_off}. Intuitively, when $S$ is slowly varying in intervals of length $\tau$ and $T-\tau$, the probability of $E_k=0$ depends, for any $k$, only on the length of intercontact times, starting from a point in time that is uniformly distributed in an ON or OFF interval, while all dependencies on previous events can be neglected. The joint probability in Equation~\ref{eq:n_conditioned_with_events} becomes the product of the marginal probabilities, and the marginal probabilities admit simple expressions. So we obtain Theorem~\ref{theo:pdf_n_approx}.}

\begin{theorem}[PMF of $N$]\label{theo:pdf_n_approx}
When the PDF of intercontact times $S_i$ is slowly varying in any interval of length $\max \{ \tau, T-\tau\}$,  the probability mass function of $N$ can be approximated as follows:
\begin{equation}\label{eq:pdf_n_approx}
\left\{ \begin{array}{lr}
\mathbb{P}\{N=1\} = g &\\
\mathbb{P}\{N=k\} = (1-g)(1-p)^{k-2}p, & k\geq 2
\end{array}\right.
\end{equation}
where  $g=\sum_{n_1=0}^{\infty}  \mathbb{P}\left(Z^{ON} + S \in \mathcal{I}^{ON}_{n_1}\right)$, $p=\sum_{n_2=1}^{\infty}  \mathbb{P}\left(Z^{OFF} + S \in \mathcal{I}^{ON}_{n_2}\right)$, and $Z^{ON} \sim Unif(0, \tau)$, $Z^{OFF} \sim Unif(0, T-\tau)$ as shown in Lemma~\ref{lemma:uniformity} and Lemma~\ref{lemma:uniformity_off}.
\end{theorem}

\begin{proof}
Let us provide an intuitive explanation for this result (the detailed proof can be found in~\cite{biondi:what_you_lose:tr}). Focusing on the first event ($E_1$), it is easy to see that $\mathbb{P}\{E_1=1\}=\mathbb{P}\{N=1\}$. We synthetically denote this quantity as $g$. The value of $g$ can be computed after noticing that it corresponds to the probability of detecting a contact knowing that the previous contact had been detected (i.e., that it happened in an ON interval). Please note that, thanks to the slowly varying assumption, we can assume that the previous contact took place $Z^{ON}$ seconds after the beginning on the ON interval in which it was detected. 
\addCB{For $\mathbb{P}\{N=k\}$, the first $k-1$ contacts are missed. The first one starts during an ON interval, and therefore the probability that it is missed is $1-g$. The next $k-2$ start during an OFF interval. Conceptually similar to $g$, $p$ denotes the probability of detecting a contact when the previous one starts during an OFF (instead of an ON) period. Thus, the probability of the $k-2$ events is $(1-p)^{k-2}$. Finally, we detect the $k$-th contact, which happens with probability $p$.}
%
\end{proof}

Starting from the distribution of $N$ in Theorem~\ref{theo:pdf_n_approx}, we can easily derive $N$'s first two moments and its coefficient of variation, which will be later used to compute the moments of $\tilde{S}$.

\begin{corollary}[Moments of $N$]
The first two moments and the coefficient of variation of the approximate PMF of $N$ are $\mathbb{E}[N] = \frac{1-g+p}{p}$, $\mathbb{E}[N^2] = \frac{-g (p+2)+p^2+p+2}{p^2}$, $cv_N^2 = \frac{(1-g) (g-p+1)}{(-g+p+1)^2}$.
\end{corollary}

In Section~\ref{sec:ict_negligible_pdf_n_validation} we will see how to derive $g$ and $p$ for two popular distributions of intercontact times (exponential and Pareto). This basically comes down to deriving quantities $Z^{ON} + S$ and $Z^{OFF} + S$, which are the sum of a uniform random variable and~$S$, and then to compute, either symbolically or numerically, the infinite sums. However, the expressions for $g$ and $p$ can be further simplified, as shown in Corollary~\ref{coro:N_geometric} below (whose proof can be found in~\cite{biondi:what_you_lose:tr}), again relying on the slowly varying property of $f_S$.

\begin{corollary}[$N$ as geometric r.v.]\label{coro:N_geometric}
When $f_{S}$ is a constant in intervals of length $\max\{\tau, T-\tau\}$, it holds that $g$ and $p$ are both approximately equal to $\frac{\tau}{T}$, hence the distribution of $N$ becomes approximately Geometric with parameter $\frac{\tau}{T}$ and its moments and coefficient of variation become $\mathbb{E}[N] = \frac{T}{\tau}$, $\mathbb{E}[N^2] = \frac{T (2 T-\tau )}{\tau^2}$, $cv_N^2 = \frac{T-\tau}{T}$.
%
\end{corollary}

Note that there is a subtle difference between Corollary~\ref{coro:N_geometric} and Theorem~\ref{theo:pdf_n_approx}. The former needs that $f_{S}$ is constant  in intervals of length $\max\{\tau,T-\tau\}$, while the latter only requires that $f_{S}$ is slowly varying in these intervals. Therefore, in principle results from Theorem~\ref{theo:pdf_n_approx} are less approximate than those from Corollary~\ref{coro:N_geometric}, although in practice the two conditions on $f_{S}$ are often equivalent.

In the related literature~\cite{zhou2013energy,wang2007adaptive,kouyoumdjieva2016impact}, the probability to skip a contact is often used to control and optimise the duty cycle, by making sure that the number of missed contact remains below a given threshold. This probability can be also computed with our model, and it simply corresponds to $\mathbb{P}(N\ge2) = 1- g$.



\subsubsection{Validation of $\mathbb{P}\{N=k\}$}
\label{sec:ict_negligible_pdf_n_validation}

In this section we validate our approximation for the PMF of $N$ against simulation results. We set $T$ to $100s$, which, as discussed in Section~\ref{sec:dc_process}, was found to be a good trade-off between device discoverability and energy consumption in~\cite{trifunovic2014adaptive}. Then, we explore the parameter space assuming that intercontact times are either exponential or Pareto distributed. Monte-Carlo simulations are performed drawing intercontact time samples from these distributions and filtering them according to the reference duty cycling process (considering two cases, $\tau=20$ and $\tau=80$, corresponding to $20\%$ and $80\%$ duty cycling). The number of contacts filtered out between two detected contacts gives us one sample of $N$. We collect $100,000$ samples for $N$. When all samples for $N$ are collected, we use them to compute the empirical PMF, which is compared against the analytical predictions in Theorem~\ref{theo:pdf_n_approx} and Corollary~\ref{coro:N_geometric}.

We first consider the case of exponential intercontact times and we explore four values for the rate $\lambda$: $10$, $0.1$, $0.01$, and $0.001 s^{-1}$. When real intercontact times are exponential, $g$ and $p$ in Theorem~\ref{theo:pdf_n_approx} can be derived in closed-form:
\begin{eqnarray}
g &=& \frac{e^{-\lambda  \tau }+\lambda  \tau +\frac{e^{-\lambda  \tau } \left(e^{\lambda  \tau }-1\right)^2}{e^{\lambda  T}-1}-1}{\lambda  \tau } \nonumber \\
p &=& \frac{-e^{\lambda  \tau }-e^{\lambda  (T-\tau )}+e^{\lambda  T}+1}{\lambda  \left(e^{\lambda  T}-1\right) (T-\tau )} .
\end{eqnarray}
%

Because of the sufficient condition $\lambda T \ll 1$, we expect the model to provide very good approximations when $\lambda = 0.001s^{-1}$ ($\lambda T = 0.1 \ll 1$), reasonable approximations  when $\lambda = 0.01s^{-1}$ ($\lambda T = 1$), and discrepancies  when $\lambda \ge 0.1s^{-1}$ (because $\lambda T >1$). This is exactly what we observe in Figures~\ref{fig:PMF_n_approx_exp_20}-\ref{fig:PMF_n_approx_exp_80}, except that the model is pretty close to simulation results also when $\lambda=0.1s^{-1}$ and it predicts even better for $\lambda=10 s^{-1}$. To understand why, we need to observe closely what happens when $\lambda T \gg 1$. Specifically, in this case many contacts happen within the same interval (either ON or OFF). This means that there are long sequences of contacts detected one after the other one, but, as soon as the first contact falls in an OFF interval, there will also be many missed contacts. 
\addCB{This results in a distribution of $N$ that is almost exclusively concentrated on $N=1$, with a small peak in the tail (e.g., for $\lambda=10s^{-1}$ the tail of $N$~accumulates at $800s$, which is exactly the average number of missed contacts when the OFF interval has length $80s$). Theorem~\ref{theo:pdf_n_approx} is able to predict the accumulation at $N=1$ (see Figure~\ref{fig:PMF_n_approx_exp_tail} for $\lambda=10s^{-1}$ and $0.1s^{-1}$). The inaccuracy in predicting the peak in the tail is not very critical, as anyway the total mass of probability under these peaks is typically very limited (we will show in Section~\ref{sec:ict_negligible_stilde_validation} that it has no practical effect in predicting~$\tilde{S}$).}
%
As far as the Geometric approximation in Corollary~\ref{coro:N_geometric} is concerned, we observe that it works very well when  $\lambda T  \ll 1$, but it degrades quickly when this assumption does not hold, especially for smaller values of $N$ (Figures~\ref{fig:PMF_n_approx_exp_20}-\ref{fig:PMF_n_approx_exp_80}). We have already anticipated the reason for this behaviour. In fact, Theorem~\ref{theo:pdf_n_approx} exploits the uniformity of contact displacement in the intervals in which they take place but then uses the actual functional form of $f_S$. For this reason, when the slowly varying assumption does not hold, predictions are generally more accurate with Theorem~\ref{theo:pdf_n_approx} than with Corollary~\ref{coro:N_geometric}. Finally, as shown in~\cite{biondi:what_you_lose:tr}\addCB{, where we fully develop the case of exponential intercontact times, the better performance of the Geometric approximation for the tail is an artefact of the fact that, in case of exponential intercontact times, the tail decays as $(T-\frac{\tau}{T})^{k-1}$. However, this is paid with a very significant inaccuracy in predicting the body of the distribution (and specifically N=1), which is way more important given the accumulation of probability mass around that value.}

\begin{figure}[h]
\begin{center}
\subfigure[$\tau=20s, T=100s$\label{fig:PMF_n_approx_exp_20}]{\includegraphics[scale=0.45]{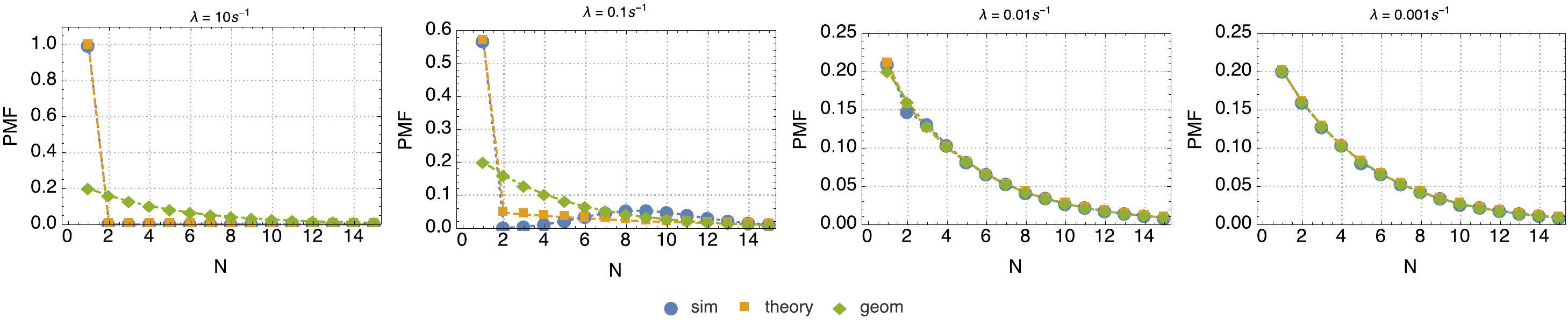}}
\subfigure[$\tau=80s, T=100s$\label{fig:PMF_n_approx_exp_80}]{\includegraphics[scale=0.45]{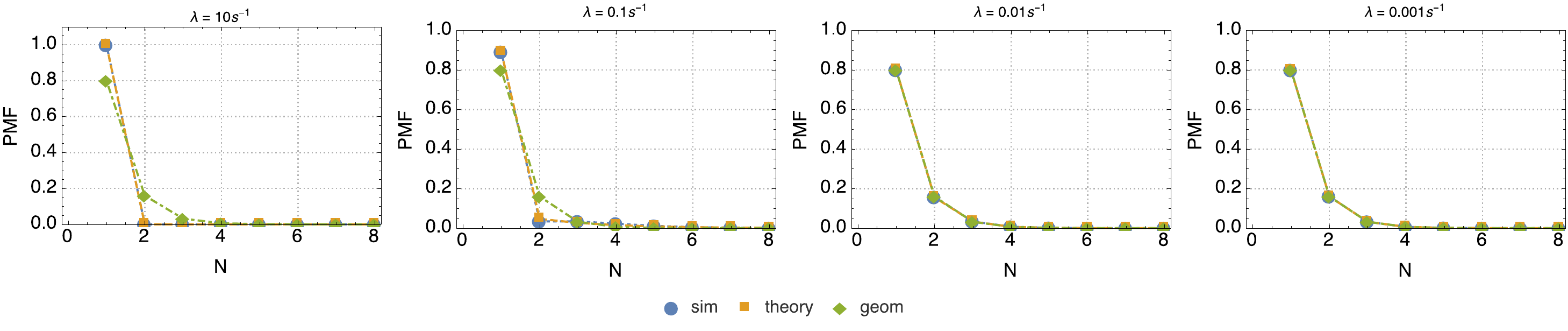}}
\subfigure[Tail behaviour, $\tau=20s$\label{fig:PMF_n_approx_exp_tail}]{\includegraphics[scale=0.45]{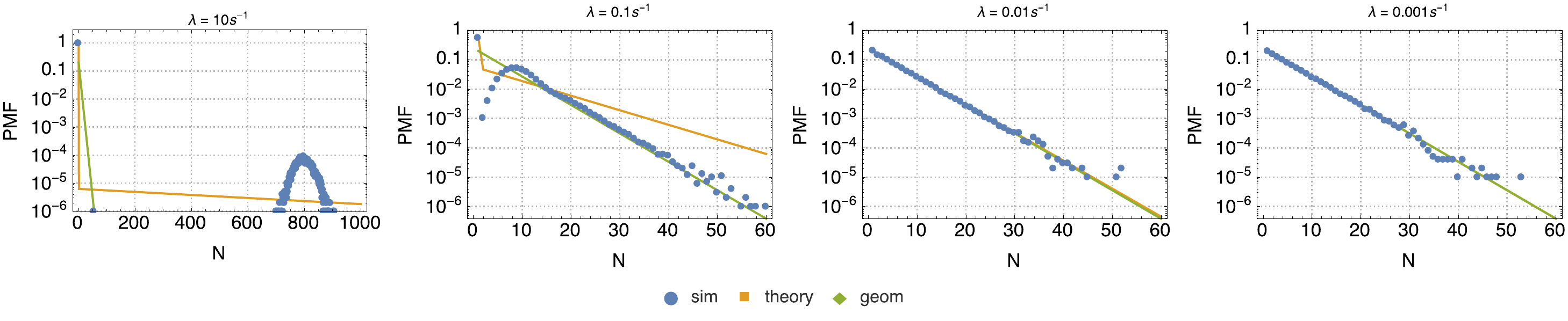}}
\caption{PMF of N: predictions of the approximate model VS empirical distribution - Exponential}
\label{fig:PMF_n_approx_exp}
\end{center}
\end{figure}

For the Pareto case, the sufficient condition for $f_S$ to be slowly varying in any interval of length~$T$ is $\frac{T}{\mathbb{E}[S](\alpha-1)} \ll 1$ or equivalently $\frac{T}{b} \ll 1$. Again, we take $T=100s$ and two different values for $\tau$ ($20s$ and $80s$). Since the Pareto distribution is characterised by two parameters,  here we have an additional degree of freedom. So, for the purpose of validation, we have selected $\alpha = 1.01$ for the Pareto, which is right after the threshold $\alpha=1$ for the convergence of the expectation~\cite{boldrini2015stability}. With $\alpha$ fixed, we have selected $b$ values in order to extensively explore the parameter space. Thus, we have chosen $b \in \{1, 10, 100, 1000\}$. Simulations are performed as discussed for the exponential case. Also in the Pareto case $g$ and $p$ in Theorem~\ref{theo:pdf_n_approx} can be derived in closed form:
\begin{eqnarray}
\scriptstyle g & \scriptstyle =& \scriptstyle \frac{T^{-\alpha } \left(T b^{\alpha } \left(\zeta \left(\alpha -1,\frac{b+T-\tau }{T}\right)+\zeta \left(\alpha -1,\frac{b+T+\tau }{T}\right)\right)-2 T^{\alpha +1} \left(\frac{b}{T}\right)^{\alpha } \zeta \left(\alpha -1,\frac{b+T}{T}\right)+T^{\alpha } \left((\alpha -1) \tau +(b+\tau ) \left(\frac{b}{b+\tau }\right)^{\alpha }-b\right)\right)}{(\alpha -1) \tau } \nonumber \\
\scriptstyle p & \scriptstyle =& \scriptstyle -\frac{T \left(\frac{b}{T}\right)^{\alpha } \left(-\zeta \left(\alpha -1,\frac{b+T+\tau }{T}\right)+\zeta \left(\alpha -1,\frac{b+T}{T}\right)+\zeta \left(\alpha -1,\frac{b+T-(T-\tau)+\tau }{T}\right)-\zeta \left(\alpha -1,\frac{b+T-(T-\tau)}{T}\right)\right)}{(\alpha -1) (T-\tau)}, 
\end{eqnarray}
where $\zeta(\cdot, \cdot)$ denotes the Hurwitz zeta function. 

Given the condition $\frac{T}{b} \ll 1$, we expect discrepancies between the model and simulation results for the first two values of $b$, a reasonable approximation for $b=100$, and an accurate prediction for $b=1000$.  From Figure~\ref{fig:PMF_n_approx_par} we observe that predictions are actually very good also in those cases in which they were expected to be less accurate. This happens for the same reason as for the exponential case. Also similarly to the exponential case, the predictions of the Geometric approximation in Corollary~\ref{coro:N_geometric} become quickly worse as $b$ shrinks. 

\begin{figure}[h]
\begin{center}
\subfigure[$\alpha=1.01, \tau=20, T=100$\label{fig:PMF_n_approx_par_20}]{\includegraphics[scale=0.45]{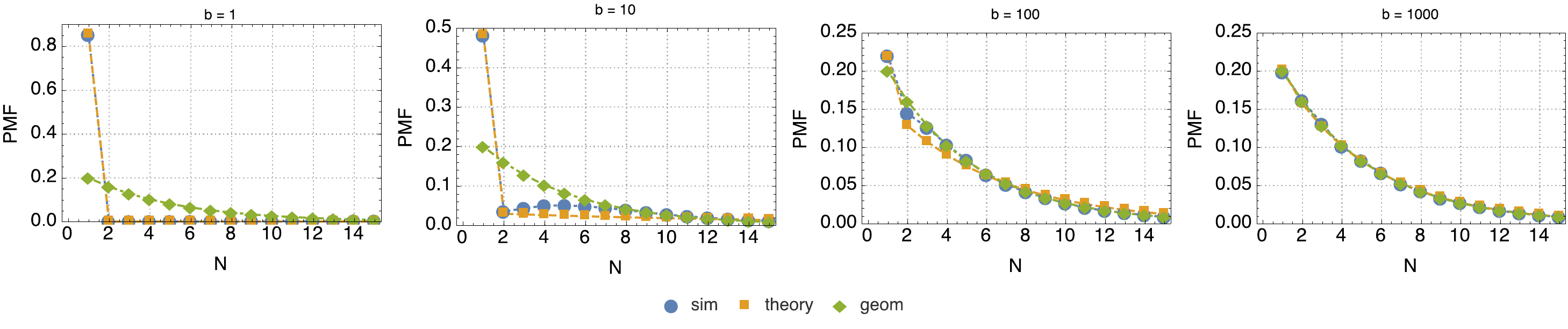}}
\subfigure[$\alpha=1.01, \tau=80, T=100$\label{fig:PMF_n_approx_par_80}]{\includegraphics[scale=0.45]{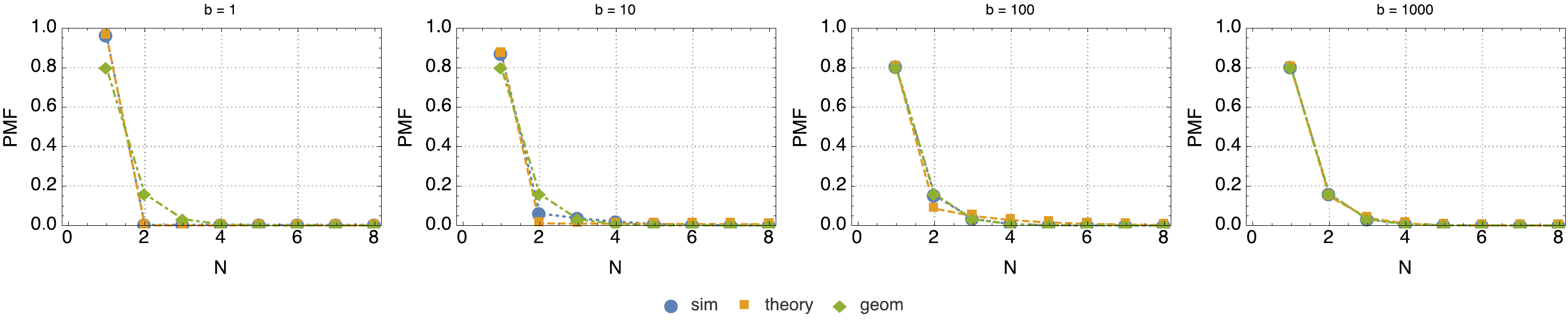}}
\subfigure[Tail behaviour, $\tau=20$\label{fig:PMF_n_approx_par_tail}]{\includegraphics[scale=0.45]{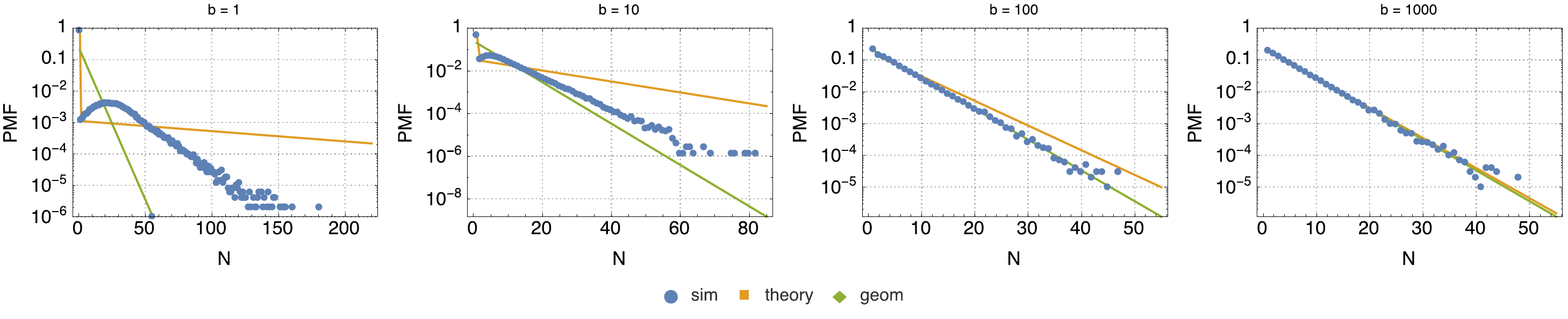}}
\caption{PMF of N: predictions of the approximate model VS empirical distribution - Pareto}
\label{fig:PMF_n_approx_par}
\end{center}
\end{figure}

Summarising, for both exponential and Pareto intercontact times, the predictions of Theorem~\ref{theo:pdf_n_approx} are generally very accurate, in particular for the body of the PMF of~$N$. As for the tail of the distribution, it is accurately captured only when the slowly varying conditions are satisfied. In Section~\ref{sec:ict_negligible_stilde_validation} we will see whether this inaccurate prediction affects the results for the measured intercontact times $\tilde{S}$. Finally, the Geometric approximation, despite being very convenient for its simplicity, is safe to use only when the slowly varying conditions are satisfied. When conditions are not satisfied, the Geometric approximation is not able to capture the behaviour of the body of the distribution and, more in general, the variability of $N$.


\subsection{Measured intercontact times}
\label{sec:ict_negligible_detected_ict}

Exploiting the results in the previous section, here we discuss how to compute the first and second moment of the measured intercontact time $\tilde{S}$ for a generic node pair. \addCB{We know that the relation between $S$ and $\tilde{S}$ is $\tilde{S} = \sum_{i=1}^{N} S_i$.} Thus, $\tilde{S}$ is a random sum of random variables, and we can exploit well-known properties to compute its first and second moment. With the results in Proposition~\ref{prop:moments_of_tilde_S} below we can exploit commonly used representations of positive random variables as hyper/hypo-exponential random variables. The rationale of this approach is to use the moments of $\tilde{S}$ to fit a hyper- or a hypo-exponential distribution (depending on whether $cv_{\tilde{S}}^2$ is greater or smaller than~$1$, respectively) and then use this representation of $\tilde{S}$ to derive important metrics (e.g., the delay and number of hops) characterising the performance of networking protocols in the opportunistic network.
\begin{proposition}\label{prop:moments_of_tilde_S}
The first and second moment of $\tilde{S}$ are given by the following:
\begin{eqnarray}
\mathbb{E}[\tilde{S}] &=& \mathbb{E}[N] E[S] \\
\mathbb{E}[\tilde{S}^2] &=& \mathbb{E}[N^2]\mathbb{E}[S]^2 + \mathbb{E}[N]\mathbb{E}[S^2]-\mathbb{E}[N]\mathbb{E}[S]^2 \\
cv_{\tilde{S}}^2 &=& \frac{cv_S^2}{\mathbb{E}[N]} + cv_N^2. \label{eq:cv_N_negligible}
\end{eqnarray}
\end{proposition}

In the following we will use these formulas to get some general insights on the possible behaviour of $\tilde{S}$ depending on the distribution of $S$ and the duty cycle $\Delta$. By ``behaviour" we mean whether the distribution of $\tilde{S}$ exhibits a hypo-exponential, exponential or hyper-exponential nature. When real intermeeting times are exponential and $f_S$ is slowly varying according to Lemma~\ref{lemma:uniformity} and Lemma~\ref{lemma:uniformity_off}, we found that measured intercontact times feature a coefficient of variation that is approximately equal to $1$, from which it follows that $\tilde{S}$ can be approximated with an exponential distribution. In Lemma~\ref{lemma:stilde_neg_exp} below we summarise this result and we also provide an expression for the rate of $\tilde{S}$ in this case.

\begin{lemma}\label{lemma:stilde_neg_exp}
When intercontact times $S_i$ are exponential with rate $\lambda$ and condition $\lambda T \ll 1$ holds, the measured intercontact times are again exponential but with rate $\lambda \Delta$.
\end{lemma}

\begin{proof}
When $\lambda T \ll 1$ we can apply Corollary~\ref{coro:N_geometric}. Thus, recalling that $\mathbb{E}[S] = \frac{1}{\lambda}$ and $cv_S^2 =1$, and substituting them into the equations in Proposition~\ref{prop:moments_of_tilde_S}, we obtain $\mathbb{E}[\tilde{S}] = \frac{1}{\lambda \Delta}$ and $cv_{\tilde{S}}^2 = 1$, which corresponds to an exponentially distributed random variable with rate $\lambda \Delta$. Please note we had already obtained this result in~\cite{biondi2014duty} using a more complex model and derivation.
\end{proof}


It has been proved in the literature~\cite{chaintreau2007impact,boldrini2015stability} \addCB{that the expected delay of routing protocols may not converge in opportunistic networks with Pareto intercontact times depending on the value of their shape parameter~$\alpha$, which determines the behaviour of the tail. Lemma~\ref{lemma:s_tilde_neg_paretotail} below tells us that if there are convergence issues with~$S_i$, the same problems will show up also with the measured intercontact times, since they decay as a Pareto with the same exponent. }

\begin{lemma}\label{lemma:s_tilde_neg_paretotail}
When intercontact times are Pareto with exponent $\alpha$, the CCDF of $\tilde{S}$ decays as a Pareto random variable with exponent $\alpha$.
\end{lemma}

\begin{proof}
Since $\tilde{S} = \sum_{i=1}^{N} S_i$, we can write the following:
\begin{equation}
\mathbb{P}(\tilde{S} > x) = \sum_{k=1}^{\infty} \mathbb{P}\{N=k\} \mathbb{P}(S_1 + \cdots + S_k > x).
\end{equation}
It has been shown~\cite{ramsay2006distribution} that, when $S_i$ are i.i.d. Pareto random variables, $\mathbb{P}(S_1 + \cdots + S_k > x)$ goes as $L(x) k \left(\frac{b}{x}\right)^{\alpha}$ when $x \rightarrow \infty$, where $L(x)$ is a slowly varying (in the sense of~\cite{bingham1989regular})
 function of $x$. Using this relationship, we can rewrite the above equality as:
\begin{eqnarray}
\mathbb{P}(\tilde{S} > x) &\sim & L(x) \left(\frac{b}{x}\right)^{\alpha} \sum_{k=1}^{\infty} k \mathbb{P}\{N=k\} \\
& \sim & L(x) \left(\frac{b}{x}\right)^{\alpha} \mathbb{E}[N] \\
& \sim & \left(\frac{b}{x}\right)^{\alpha}
\end{eqnarray}
\end{proof}
%

Now, in Lemma~\ref{lemma:stilde_neg_general} (whose proof can be found in~\cite{biondi:what_you_lose:tr}) we want to answer a more general question: can duty cycling transform a hypo-exponential intercontact time into a hyper-exponential measured intercontact time, and vice versa? The answer is yes, and whether this happens or not depends on the values of $cv_S^2, g$ and $p$. This result has a serious implication for opportunistic networks: if exponential intercontact times can turn into hyper-exponential measured intercontact times, delay convergence issues may show up also in apparently safe exponential mobility scenarios if a duty cycling policy is deployed.

\begin{lemma}\label{lemma:stilde_neg_general}
When $f_S$ is slowly varying in intervals of length $\max \{ \tau, T-\tau\}$, the measured intercontact times behave according to Table~\ref{tab:hypo_hyper_regions}, where $\xi(g,p) = 1 - \frac{2g}{p} + \frac{2}{1-g+p}$ and $\omega(g) = \frac{3}{2}(g-1)+\frac{1}{2}\sqrt{9-10g+g^2}$.
\end{lemma}

\begin{table}[ht]
\caption{Summary of hypo-exponential/hyper-exponential regions depending on $cv_S^2, g,p$}
\fontsize{8.0}{9.0}\selectfont
{\renewcommand{\arraystretch}{1.5}
\begin{tabular}{c||c|c|}
& $\tilde{S}$ hypo-exponential & $\tilde{S}$ hyper-exponential \\[3pt] \hline \hline
S hyper-exponential $(cv_S^2 > 3)$ & - & always \\ \hline
\multirow{2}{*}{S hyper-exponential $(1 < cv_S^2 <= 3)$} & $p > g \wedge  cv_S^2 < \xi(g,p) $ & $g \ge p$\\
& & $p > g \wedge  cv_S^2 > \xi(g,p) $ \\ \hline
\multirow{2}{*}{S hypo-exponential} & $p \ge g$ & $p < \omega(g)$ \\ 
& $g > p > \omega(g) \land cv_S^2 < \xi(g,p)$ & $g > p > \omega(g)\land cv_S^2 > \xi(g,p) $  \\ \hline
\hline
\end{tabular}}
\label{tab:hypo_hyper_regions}
\end{table}

Despite their apparent complexity, the conditions in Table~\ref{tab:hypo_hyper_regions} have an intuitive explanation. Take for example the important case where hypo-exponential intercontact times can turn into hyper-exponential measured intercontact times, corresponding to the bottom right entry of Table~\ref{tab:hypo_hyper_regions}. As worked out in~\cite{biondi:what_you_lose:tr}, condition $p < \omega(g)$ corresponds to the case $cv_N^2 > 1$, which, according to Equation~\ref{eq:cv_N_negligible} is enough for the coefficient of variation of $\tilde{S}$ to be greater than $1$. In practice, what happens with $cv_N^2 > 1$ is that both very small and very large values of $N$ are possible, hence $\tilde{S}$ will be a combination of the sums of a small number of intercontact times $S$ and the sums of a large number of intercontact times $S$, thus introducing a lot of variability in the distribution of $\tilde{S}$. Vice versa, when condition $p < \omega(g)$ does not hold, $cv_S^2$ starts to weigh in in Equation~\ref{eq:cv_N_negligible} and can tilt the balance in favour of hyper-exponentiality depending on how it relates to $\xi(g,p)$ (cfr second condition in the bottom right entry of Table~\ref{tab:hypo_hyper_regions}).

Finally, it is interesting to note from Table~\ref{tab:hypo_hyper_regions}, that when $g=p$ (which we know happens under the conditions of Corollary~\ref{coro:N_geometric}) $\tilde{S}$ inherits the behaviour of $S$: if $S$ is hyper-exponential then $\tilde{S}$ is hyper-exponential, and vice versa. In Section~\ref{sec:ict_negligible_stilde_validation} we use this result to assess the behaviour of $\tilde{S}$. Hence, with Corollary~\ref{coro:N_geometric} we would never observe a change in the behaviour of $\tilde{S}$.

\subsubsection{Validation of $\tilde{S}$}
\label{sec:ict_negligible_stilde_validation}

In order to validate the model, we compare simulation results against theoretical predictions for $\tilde{S}$. Specifically, we generate the measured intercontact times via Monte Carlo simulation, similarly to what we have done for the validation of $N$. The theoretical predictions are obtained as discussed in the previous section.

We start our validation with the exponential case and in Figure~\ref{fig:s_tilde_exp} we plot the distribution of $\tilde{S}$ for the same parameter values ($\lambda \in \{ 0.001s^{-1}, 0.01s^{-1}, 0.1s^{-1}, 10s^{-1}\}$) used in the validation of the PMF of~$N$. We plot in orange the distribution predicted using Theorem~\ref{theo:pdf_n_approx} and in green the distribution of an Exponential random variable with rate $\lambda \Delta$ (corresponding to the prediction of Lemma~\ref{lemma:stilde_neg_exp}, which in turn relies on Corollary~\ref{coro:N_geometric}). Since the PMF of $N$ was predicted with very good accuracy by our model in Theorem~\ref{theo:pdf_n_approx}, we also expect that the distribution of $\tilde{S}$ matches its theoretical prediction. This is indeed the case, both for $\lambda$ values such that $\lambda T \ll 1$ but also in the opposite case, when $\lambda T \ge 1$. This is due to the fact that when $\lambda$ is large (with respect to $\tau$) $\mathbb{P}(N = 1)$ becomes large, and this effect is captured by Theorem~\ref{theo:pdf_n_approx} (as we have discussed in Section~\ref{sec:ict_negligible_pdf_n_validation}). Vice versa, when $\lambda$ is large the predictions based on the Geometric approximation in Corollary~\ref{coro:N_geometric} fail to capture the high probability of event $N=1$. 

In Figure~\ref{fig:s_tilde_exp_tail} we also plot the tail of the CCDF of $\tilde{S}$. We observe that both models accurately predict the tail when $\lambda T \ll 1$ holds. When the condition does not hold, we need to apply Lemma~\ref{lemma:stilde_neg_general}. Specifically, computing $p$, $g$ and $\omega(g)$, Lemma~\ref{lemma:stilde_neg_general} predicts that $\tilde S$ is hyper-exponential, while Lemma~\ref{lemma:stilde_neg_exp} would predict a squared coefficient of variation equal to $1$. Indeed, the squared coefficient of variation measured in simulation is $142$.

We complete the validation of the results for $\tilde{S}$ by considering Pareto intercontact times (with $\alpha=1.01$, as in Section~\ref{sec:ict_negligible_pdf_n_validation}). When the slowly varying condition holds ($b=100$, $b=1000$), we observe a perfect match between simulation results and predictions (with both proposed models, as can be seen in Figure~\ref{fig:s_tilde_par_all}). When the hypotheses of the model are not met (small $b$), the model is nevertheless able to well approximate simulation results, though with a slight mismatch in the initial part of the CDF, where the proposed model underestimates the presence of small values of $\tilde{S}$. This is due to the long tail of $N$ under the model in Theorem~\ref{theo:pdf_n_approx}, which we had observed in Figure~\ref{fig:PMF_n_approx_par_tail}. Owing to this fact, the model overestimate the presence of large values of $N$, and this has an impact on the predictions for $\tilde{S}$. However, this effect is quite limited. The above considerations also apply to the tail of $\tilde{S}$ in Figure~\ref{fig:s_tilde_par_tail}. Here we have also plotted a curve for $(\frac{b}{x})^{\alpha}$, in order confirm that $\tilde{S}$  decays as predicted by Lemma~\ref{lemma:s_tilde_neg_paretotail}.

\begin{figure}[t]
\begin{center}
\subfigure[$\tau=20, T=100$\label{fig:s_tilde_exp_all}]{\includegraphics[scale=0.4]{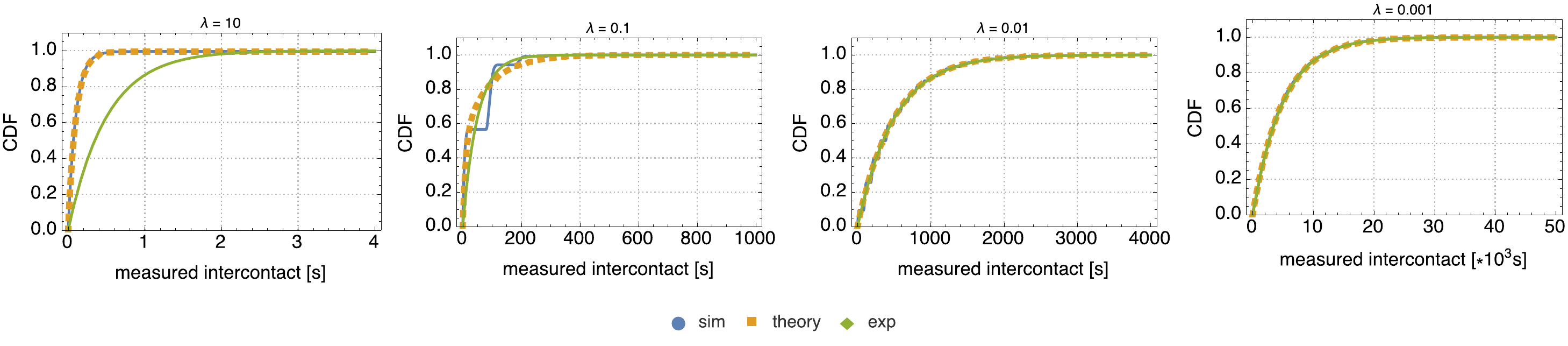}}
\subfigure[$\tau=20, T=100\label{fig:s_tilde_exp_tail}$, tail behaviour]{\includegraphics[scale=0.4]{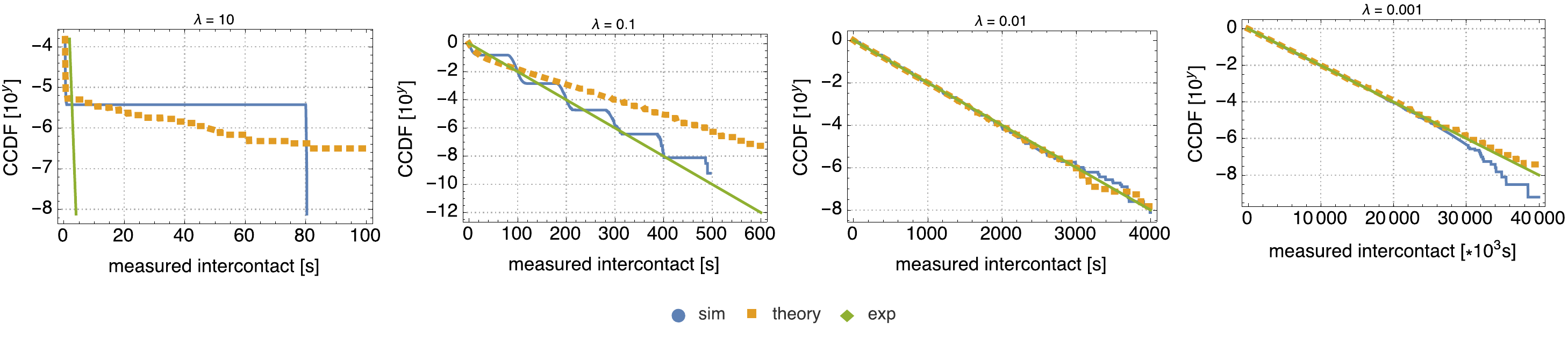}}
\caption{$\tilde{S}$: predictions of the approximate model VS empirical distribution - Exponential.}
\label{fig:s_tilde_exp}
\end{center}
\end{figure}

\begin{figure}[t]
\begin{center}
\subfigure[$\tau=20, T=100$\label{fig:s_tilde_par_all}]{\includegraphics[scale=0.4]{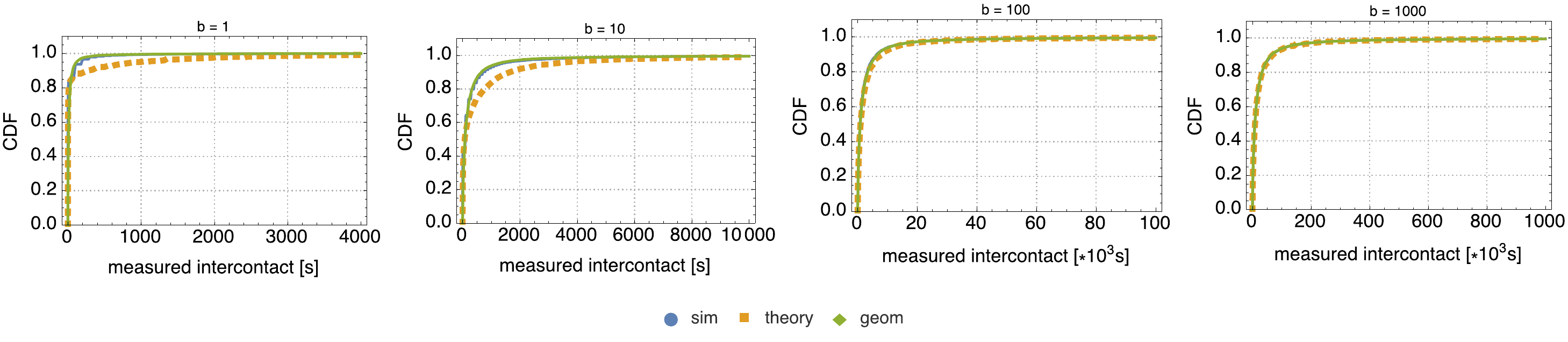}}
\subfigure[$\tau=20, T=100$, tail behaviour\label{fig:s_tilde_par_tail}]{\includegraphics[scale=0.4]{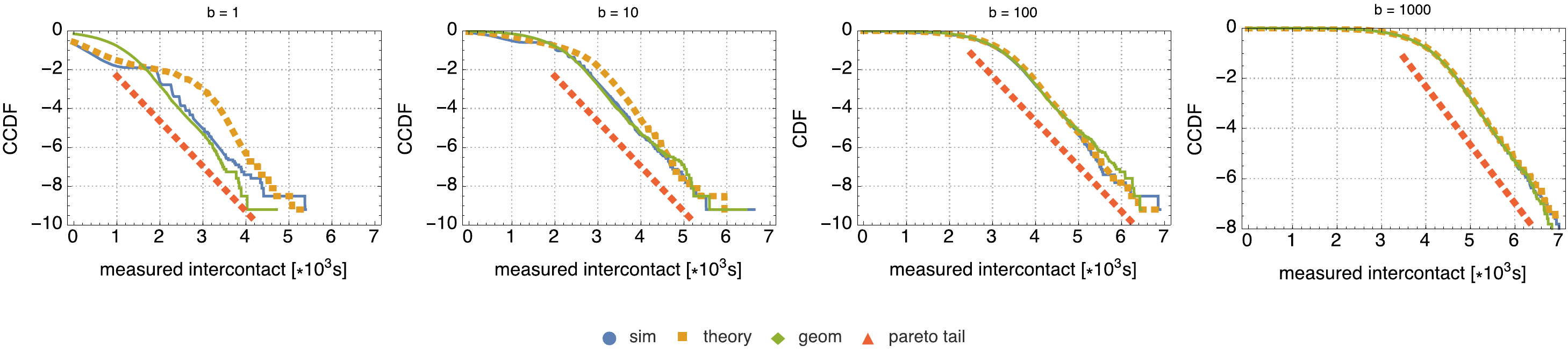}}
\caption{$\tilde{S}$: predictions of the approximate model VS empirical distribution - Pareto.}
\label{fig:s_tilde_pareto}
\end{center}
\end{figure}


\section{Detected contact process when contact duration is non-negligible}
\label{sec:ict_non_negligible}

\addCB{While this is typically a reasonable approximation, contact durations might in general not be negligible with respect to duty cycle periods. Therefore, }in Section~\ref{sec:ict_non_negligible_preliminaries} we revise the slowly varying conditions for the case of non-negligible contact duration (this boils down to extending the slowly varying assumption to the distribution of the contact duration, $C_i$).
Then in Section~\ref{sec:ict_non_negligible_measured_contact} we compute the distribution of the measured contact duration $\tilde{C}$. This measured contact duration cannot, by definition, exceed the length of an ON interval, and its distribution depends on how the real contact duration intersects with the ON interval. Building upon these results, in Section~\ref{sec:ict_non_negligible_measured_ict} we derive the measured intercontact times $\tilde{S_i}$. The main differences with respect to when contact duration is negligible is that the portions of contact duration not overlapping with ON intervals have to be included into the measured intercontact time. For both $\tilde{C}$ and $\tilde{S}$  we also provide a model for the case when nodes refrain from entering the OFF state upon contact detection, as in~\cite{trifunovic2014adaptive}. In this case, the effect of the duty cycle on the measured contact duration is generally limited, since only the initial part of the original contact is missed. As for the measured intercontact times, we do not observe anymore the contribution of the portion of contact duration not overlapping with ON intervals. Finally, in Section~\ref{sec:ict_non_negligible_validation} we validate this extended model, showing that its predictions are once again very accurate.

\subsection{Preliminaries}
\label{sec:ict_non_negligible_preliminaries}


In order to make the model tractable, similarly to what we have done in Section~\ref{sec:ict_negligible}, we need to model the measured contact process as a renewal process, this time of type alternate-renewal. The measured contact process $\{ \tilde{C}_i, \tilde{S}_i\}$ is alternate renewal if $\{ \tilde{C}_i, \tilde{S}_i\}$ are independent sequences of i.i.d. random variables. Based on the discussion for the negligible contact case, it is clear that also in this case the measured contact process is not in general renewal. 
However, using an argument similar to the one used in Lemma~\ref{lemma:uniformity} it is possible to prove (Lemma~\ref{lemma:uniformity_s_nonneg}) that also in this case we can assume, approximately, the independence for $\{ \tilde{C}_i, \tilde{S}_i\}$, provided that the PDF of $S_i$ is slowly varying in intervals of length $\max \{T-\tau, \tau\}$.  From this, it follows that the displacement of the beginning of a real contact in the ON/OFF interval in which it takes place is uniformly distributed in that interval. We omit the proof of this result as it would be a repetition of the same concepts discussed in Lemma~\ref{lemma:uniformity}.  

\begin{lemma}[$S_i$ slowly varying]\label{lemma:uniformity_s_nonneg}
When $f_S$ is a slowly varying function in any interval of length~$\max \{T-\tau, \tau\}$, the \emph{measured} contact process can be approximated as an alternating renewal process $\{ \tilde{C}_i, \tilde{S}_i\}$ and the displacement of the \emph{beginning} of a contact within its ON (or OFF) interval can be approximated as uniformly distributed in that interval.
\end{lemma}

\addCB{For our derivations, it may also be convenient to assume that not only the beginning of a real contact but also its end is approximately uniformly distributed in the interval in which it takes place. Thus, if needed, we will leverage the additional assumption that the PDF of $C_i$ is slowly varying in intervals of length $\max \{T-\tau, \tau\}$ (note, though, that this assumption is not needed for the independence of $\{ \tilde{C}_i, \tilde{S}_i\}$).}

\begin{corollary}[$S_i$ and $C_i$ slowly varying]\label{coro:uniformity_sc_nonneg}
When \emph{both} $f_S$ and $f_C$ are slowly varying functions in any interval of length~$\max \{T-\tau, \tau\}$, the displacement of the \emph{beginning (end)} of a contact within its ON (OFF) interval can be approximated as uniformly distributed in that interval.
\end{corollary}

Exploiting the fact that the measured contact process is approximately renewal under the conditions in Lemma~\ref{lemma:uniformity_s_nonneg}, we can again focus on what happens between two detected contacts. Differently from the previous case with negligible contact duration, the detected real contact corresponding to the measured contact can now start also in an OFF interval. For this reason, operating a shift of index as we have done in the previous section, the starting point of our analysis will be the interval $[0, T]$ in which we know a detected contact has taken place, where its OFF interval is $[0, T-\tau]$ and its ON interval $[T-\tau, T]$. All following OFF (ON) intervals will be of type $[(i-1)T, iT-\tau]$ ($[iT-\tau, iT]$) with $i > 1$. Based on Lemma~\ref{lemma:uniformity_s_nonneg}, we also know that the displacement of the beginning of a contact in an ON interval is distributed as $Z^{ON}\sim Unif(0, \tau)$ while that of a contact in an OFF interval as $Z^{OFF}\sim Unif(0, T-\tau)$. Similarly, under Corollary~\ref{coro:uniformity_sc_nonneg}, the end of a contact in ON and OFF intervals will be displaced as $Z^{ON}$ and $Z^{OFF}$, respectively. Given the distribution of displacements, we can easily quantify (Lemma~\ref{lemma:nonneg_ps} and Corollary~\ref{coro:nonneg_pe} below) the probability of a contact beginning (ending) in an ON interval, and that of a contact beginning (ending) in an OFF interval, which we denote as $p_s^{ON}, p_e^{ON}, p_s^{OFF}, p_e^{OFF}$, respectively. 

\begin{lemma}[Probability $p_s^{ON}$/$p_s^{OFF}$]\label{lemma:nonneg_ps}
When $f_S$ is a slowly varying function in intervals of length~$\max \{T-\tau, \tau\}$, the probability $p_s^{ON}$ of a contact beginning in an ON interval and the probability  $p_s^{OFF}$ of a contact beginning in an OFF interval are equal to $\frac{\tau}{T}$ and $\frac{T-\tau}{T}$, respectively. 
\end{lemma}

\begin{proof}
When Lemma~\ref{lemma:uniformity_s_nonneg} hold, the starting point of contacts are uniformly distributed in the duty cycling interval in which they take place. Thus, the probability of having an event in an ON or OFF interval depends on the length of that interval with respect to the overall duty cycle interval.
\end{proof}

\begin{corollary}[Probability $p_e^{ON}$/$p_e^{OFF}$]\label{coro:nonneg_pe}
If also $f_C$ is slowly varying in intervals of length~$\max \{T-\tau, \tau\}$, then the probability $p_e^{ON}$ of a contact ending in an ON interval and the probability  $p_e^{OFF}$ of a contact ending in an OFF interval are equal to $\frac{\tau}{T}$ and $\frac{T-\tau}{T}$, respectively.
\end{corollary}

\begin{proof}
It follows from Corollary~\ref{coro:uniformity_sc_nonneg} that the end point of contacts are uniformly distributed in the duty cycling interval in which they take place. Hence the same argument used in the previous proof can be exploited.
\end{proof}

\subsection{Measured contact times}
\label{sec:ict_non_negligible_measured_contact}

We start our derivation from the \emph{measured contact times}, which are defined as the time intervals during which the two nodes are in contact and both in the ON state. 
Therefore, real contact times, by definition, overlap with at least one ON interval. We denote with $H$ the number of ON intervals spanned by the detected contact. \addCB{The distribution of $H$ is very important for the rest of the analysis.} The probability that $H$ takes a specific value $h$ depends on $\tau$, $T$, and on the distribution of the contact time, as intuition suggests and as we show in Lemma~\ref{lemma:pmf_h} below.

\begin{lemma}[PMF of $H$]\label{lemma:pmf_h}
The PMF of H, defined as the number of ON intervals spanned by a detected contact, is given by the following:
\begin{equation*}
\mathbb{P}\{H = 1\} = p_s^{ON} \mathbb{P}\left(Z^{ON} + C < T \right) + p^{OFF} \mathbb{P}\left(Z^{OFF} + C^{hit} < 2T -\tau\right),
\end{equation*}
\begin{multline}\label{eq:pmf_h}
\mathbb{P}\{H = h\} = p_s^{ON} \mathbb{P}\left((h-1)T < Z^{ON} + C < h T \right) + \\
+ p_s^{OFF} \mathbb{P}\left((h-1)T-\tau < Z^{OFF} + C^{hit} < hT -\tau\right), \quad \forall h >1,
\end{multline}
where $C^{hit}$ can be computed as follows:
\begin{equation}\label{eq:pdf_c_hit}
\mathbb{P}\left(C^{hit} = c\right) = f_C(c) \int_{T-\tau-c}^{T-\tau} \frac{1}{1-F_{C}(T-\tau-z)} \frac{T-\tau}{T} \textrm{d} z.
\end{equation}
\end{lemma}

\begin{proof}
The detailed proof can be found in~\cite{biondi:what_you_lose:tr}. Its rationale is to find the probability associated with the different combinations under which a given event can occur. For example, case $H=1$ occurs either when the contact starts in an ON interval and ends in the same ON interval or when it starts in an OFF interval, lasts until the next ON  interval and ends in that same ON interval. The probabilities of these two combinations of events correspond to the first and second terms in the summation for case $H=1$ in Equation~\ref{eq:pmf_h}.
\end{proof}

Exploiting the PMF of $H$, we can now derive the measured contact time in Theorem~\ref{theo:detected_ct} below. Figure~\ref{fig:combinations} shows some possible instances of the problem when $H=1$. Note that each detected contact spanning $H$ ON intervals introduces $H$ samples of measured contact duration. In the following lemma, we investigate how these samples are characterised and how frequently they occur.
\begin{theorem}[Measured contact duration]\label{theo:detected_ct}
The measured contact time can be approximated with the following expression:
{\arraycolsep=0.5pt\def\arraystretch{2.2}
\begin{equation}\label{eq:detected_ct}
\tilde{C} = \left\{ \begin{array}{lr}
C^{short} &  \frac{\mathbb{P}(H=1)}{\sum_{h=1}^{\infty} h \mathbb{P}(H=h)} \frac{\tau}{T} \frac{\mathbb{P}(Y^{ON} \le \tau)}{\mathbb{P}(Y^{ON} \le T)}\\
C^{res}  & \frac{\mathbb{P}(H=1)}{\sum_{h=1}^{\infty} h \mathbb{P}(H=h)} \frac{\tau(T-\tau)}{T} \frac{\mathbb{P}(T-\tau < Y^{OFF} \le T)}{\mathbb{P}(T-\tau < Y^{OFF} \le 2T - \tau)}\\
Z^{ON}  & \frac{\mathbb{P}(H=1) \frac{\tau}{T} \left(1- \frac{\mathbb{P}(Y^{ON} \le \tau)}{\mathbb{P}(Y^{ON} \le T)}\right)}{\sum_{h=1}^{\infty} h \mathbb{P}(H=h)} + \frac{\mathbb{P}(H\geq 2) \left(\frac{2\tau}{T}\right)}{\sum_{h=1}^{\infty} h \mathbb{P}(H=h)}\\
\tau  & \frac{\mathbb{P}(H=1) \frac{(T-\tau)}{T} \left(1-  \frac{\mathbb{P}(T-\tau < Y^{OFF} \le T)}{\mathbb{P}(T-\tau < Y^{OFF} \le 2T - \tau)} \right)}{\sum_{h=1}^{\infty} h \mathbb{P}(H=h)} + \frac{\mathbb{P}(H\geq 2) \left(\frac{2 (T-\tau)}{T}\right)+ \sum_{h=3}^{\infty}(h-2) \mathbb{P}(H=h)}{\sum_{h=1}^{\infty} h \mathbb{P}(H=h)}
\end{array}\right., \normalsize
\end{equation}
}
where the PMF of $H$ is given in Lemma~\ref{lemma:pmf_h}, $Y^{ON} = Z^{ON} + C$, $Y^{OFF} = Z^{OFF} + C$, and the distribution of $C^{short}$ and $C^{res}$ can be computed as follows:
\begin{equation}\label{eq:pdf_c_short}
\mathbb{P}\left(C^{short} = c\right) =  \left\{ \begin{array}{lr}
 \frac{1}{\tau} f_C(c) \int_c^{\tau} \frac{1}{F_C(u)} \textrm{d} u & 0 < c < \tau \\
0 & \textrm{otherwise}
\end{array}\right., 
\end{equation}
\begin{equation}\label{eq:pdf_c_res}
\mathbb{P}\left(C^{res} \le c\right) =  \left\{ \begin{array}{lr}
\frac{F_{Y^{OFF}}(c + T - \tau) - F_{Y^{OFF}}(T - \tau) }{F_{Y^{OFF}}(T) - F_{Y^{OFF}}(T - \tau)} & 0 < c < \tau \\
0 & \textrm{otherwise}
\end{array}\right.. 
\end{equation}
\end{theorem}

\begin{proof}
The complete proof is available at~\cite{biondi:what_you_lose:tr}. In the following we will provide an intuitive explanation of its derivations. We start by noting that there are basically four ways, illustrated in Figure~\ref{fig:combinations}, in which a contact can intersect with an ON interval. The contact can be fully contained in the ON interval (case $a$), partially overlapping (cases $c,d$), or completely overlapping with the ON interval (case $b$).  In the latter case the measured contact is deterministically equal to $\tau$. In the other cases, the measured contacts are equal to $C_{short}$, $Z^{ON}$, and $C_{res}$ respectively. $Z^{ON}$ is uniformly distributed in the ON interval, while quantities $C_{short}$ and $C_{res}$ are derived in~\cite{biondi:what_you_lose:tr}. 
\begin{figure}[t]
\begin{center}
\includegraphics[scale=0.4,angle=90]{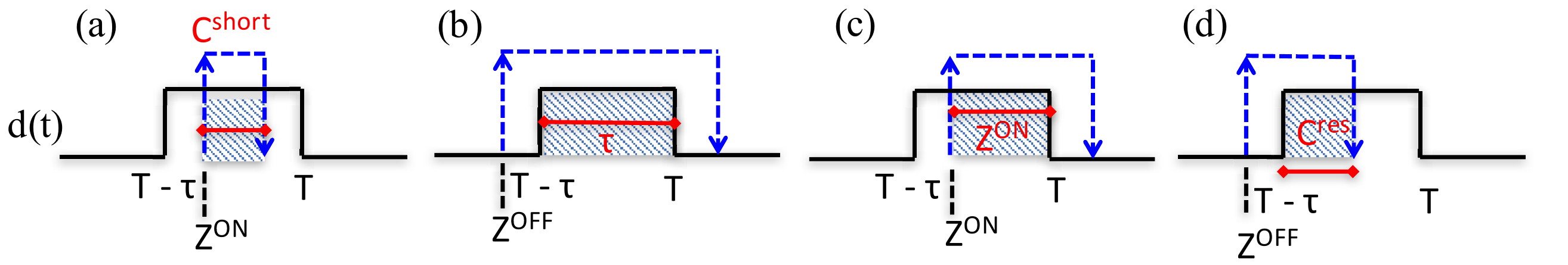}
\caption{Case $H=1$: types of overlapping.} 
\label{fig:combinations}
\end{center}
\end{figure}
Each of the four components weights differently in the distribution of the measured contact times. 
Due to lack of space, the exact derivation of these weights is left to~\cite{biondi:what_you_lose:tr}.
\end{proof}

With the above lemma we are able to fully characterise the measured contact duration in a general duty cycling scenario. When the network is sparse, i.e., the density of users is so low that the probability that a user has more than one neighbour at a given time is low as well, we can derive a useful additional result. Specifically, in some practical applications like the one in~\cite{trifunovic2014adaptive}, users are required to refrain from entering the low-power mode of the duty cycle if a new contact is detected in the current ON interval. This is done to maximise the amount of information that can be transferred between user pairs. In the general case, it can be complex to model this scenario, since a newly detected contact between node $i$ and $j$ effectively alters the joint duty cycle of $i$ and $j$ with all other nodes (because the duration of the joint ON intervals depends on this newly detected contact). However, when the network is sparse, we can ignore the effect that a detected contact between $i$ and $j$ has on the joint duty cycles with all other nodes. 
Thus, we can derive the following lemma (the proof can be found in~\cite{biondi:what_you_lose:tr}). We assume, for simplicity, that after the contact ends, the two nodes synchronise with the joint duty cycling as originally planned.

\begin{lemma}\label{lemma:detected_ct_notoff}
When nodes remain active after a contact is detected, the measured contact time can be approximated with the following expression:
\begin{equation}\label{eq:detected_ct_notoff}
\tilde{C} = \left\{ \begin{array}{lr}
C &  \frac{\tau}{T}\\
C^{res*}  & \frac{T-\tau}{T}
\end{array}\right.,
\end{equation}
where the distribution of $C^{res*}$ can be computed as follows:
\begin{equation}\label{eq:pdf_c_resstar}
\mathbb{P}\left(C^{res*} \le c\right) =  \left\{ \begin{array}{lr}
\frac{F_{Y^{OFF}}(c + T - \tau) - F_{Y^{OFF}}(T - \tau) }{1 - F_{Y^{OFF}}(T - \tau)} & 0 < c < \tau \\
0 & \textrm{otherwise}
\end{array}\right.. 
\end{equation}
and $Y^{OFF}=Z^{OFF} + C$.
\end{lemma}

In the next section, we study the measured intercontact times in order to complete the characterisation of the measured contact process. 

\subsection{The measured intercontact time}
\label{sec:ict_non_negligible_measured_ict}

Measured intercontact times $\tilde{S}_i$ are defined as the time interval between two consecutive measured contacts (which we have characterised in the previous section). 
Thus, intuitively, measured intercontact times are a composition of portions of real contact times (those that do not intersect any ON interval) and intercontact times. Actually, there is another component, which only shows up when the real contact spans more than one ON interval, which we will discuss later on. 
In order to make the derivation of $\tilde{S}$ more tractable, in the following we exploit the slowly varying approximation for~$C_i$ (Corollary~\ref{coro:uniformity_sc_nonneg}). Hence, both the start time and the end time of a real contact can be assumed uniformly distributed in the ON/OFF interval in which they take place.


We start with the simplest case, focusing on a situation in which the last detected contact overlaps with only one ON interval  (i.e., $H=1$), as shown in Figure~\ref{fig:detectedict_shortct}. The end of the previous measured contact can be the actual end point of its corresponding detected original contact (if the latter ends in an ON interval, as in Figure~\ref{fig:detectedict_shortct}(b)) or otherwise the end point of the ON interval (Figure~\ref{fig:detectedict_shortct}(a)). 
\begin{figure}[t]
\begin{center}
\includegraphics[scale=0.4,angle=90]{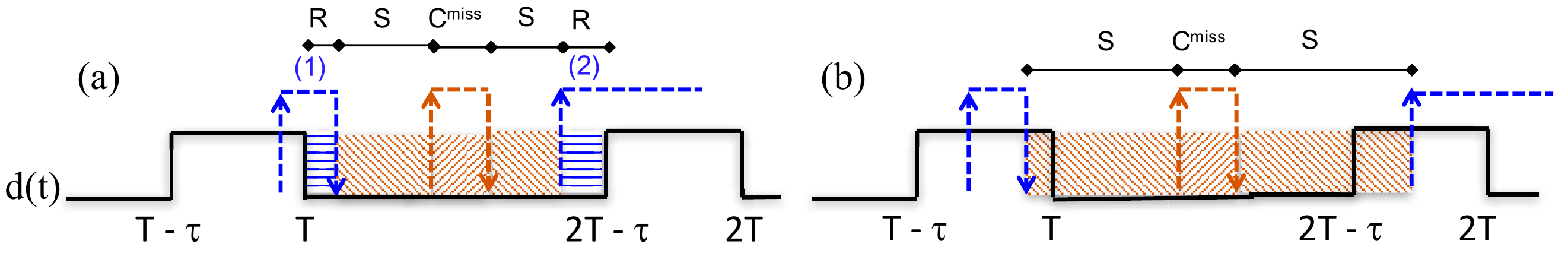}
\caption{Two examples of measured intercontact time when $H=1$.}
\label{fig:detectedict_shortct}
\end{center}
\end{figure}
In the latter case, the portion (corresponding to the blue area in Figure~\ref{fig:detectedict_shortct}(a)) of the detected contact between the end of the ON interval and the time at which the real contact ends contributes to the measured intercontact time. This quantity corresponds to the displacement of the endpoint of the real contact in the OFF interval, which is distributed as $Z^{OFF}$ if we assume $C_i$ to be slowly varying (Corollary~\ref{coro:uniformity_sc_nonneg}). \addCB{A similar line of reasoning holds for the other extreme, resulting in an additional time interval (labelled $(2)$ in the figure), again distributed as $Z^{OFF}$.}
Therefore, recalling that $p_s^{OFF}=p_e^{OFF}=\frac{T-\tau}{T}$ when both $S_i$ and $C_i$ are slowly varying in OFF intervals, these residuals can be modelled as $R$ below:
\begin{equation}\label{eq:residuals_R}
R = \left\{ \begin{array}{lr}
Z^{OFF} & \textrm{with prob. } \frac{T-\tau}{T}\\
0 &  \textrm{with prob. } \frac{\tau}{T}
\end{array}\right.
\end{equation}
Let us now focus on what happens between the two detected contacts. It is easy to see (Figure~\ref{fig:detectedict_shortct}) that, regardless of when detected contact times start and end, the central part of the measured intercontact time contains a certain number $N$ of intercontact times~$S_i$ (this is analogous to the negligible contacts case). In addition, this central part can contain \emph{missed contacts}, i.e., contacts that do not overlap with any ON interval. The probability density of these missed contacts (which we call $C_i^{miss}$) can be written as $\mathbb{P}(C_i^{miss}=c) = \mathbb{P}(C_i=c | Z^{OFF} + C_i < T-\tau)$, hence it can be obtained similarly to $C^{short}$ in Theorem~\ref{theo:detected_ct}. Putting together our observations, $\tilde{S}_i$ is equal to $R + \sum_{i=1}^{N} S_i + \sum_{i=2}^{N} C_i^{miss} + R$ when $H=1$.

When $H \ge 2$, contacts are long and span more than one ON interval. However, these long contacts cannot be used in their entirety for communication, unless under the assumptions of Lemma~\ref{lemma:detected_ct_notoff}. Specifically, their portion overlapping with OFF intervals cannot be used. Hence, each long contact is split into smaller measured contacts separated by what we call \emph{pseudo-intercontact times}, i.e. measured intercontacts of length $T-\tau$, as shown in Figure~\ref{fig:detectedict_longct}. 
What happens after the end of the long detected contact is exactly the same as what we discussed for case $H=1$: there will be a sequence of intercontact times $S_i$ and missed contacts $C_i^{miss}$, whose number depends on how many contacts are missed before the next one is detected. This corresponds to formula $R + \sum_{i=1}^{N} S_i + \sum_{i=2}^{N} C_i^{miss} + R$, which we have derived in the previous section. 

\begin{figure}[ht]
\begin{center}
\includegraphics[scale=0.5,angle=-90]{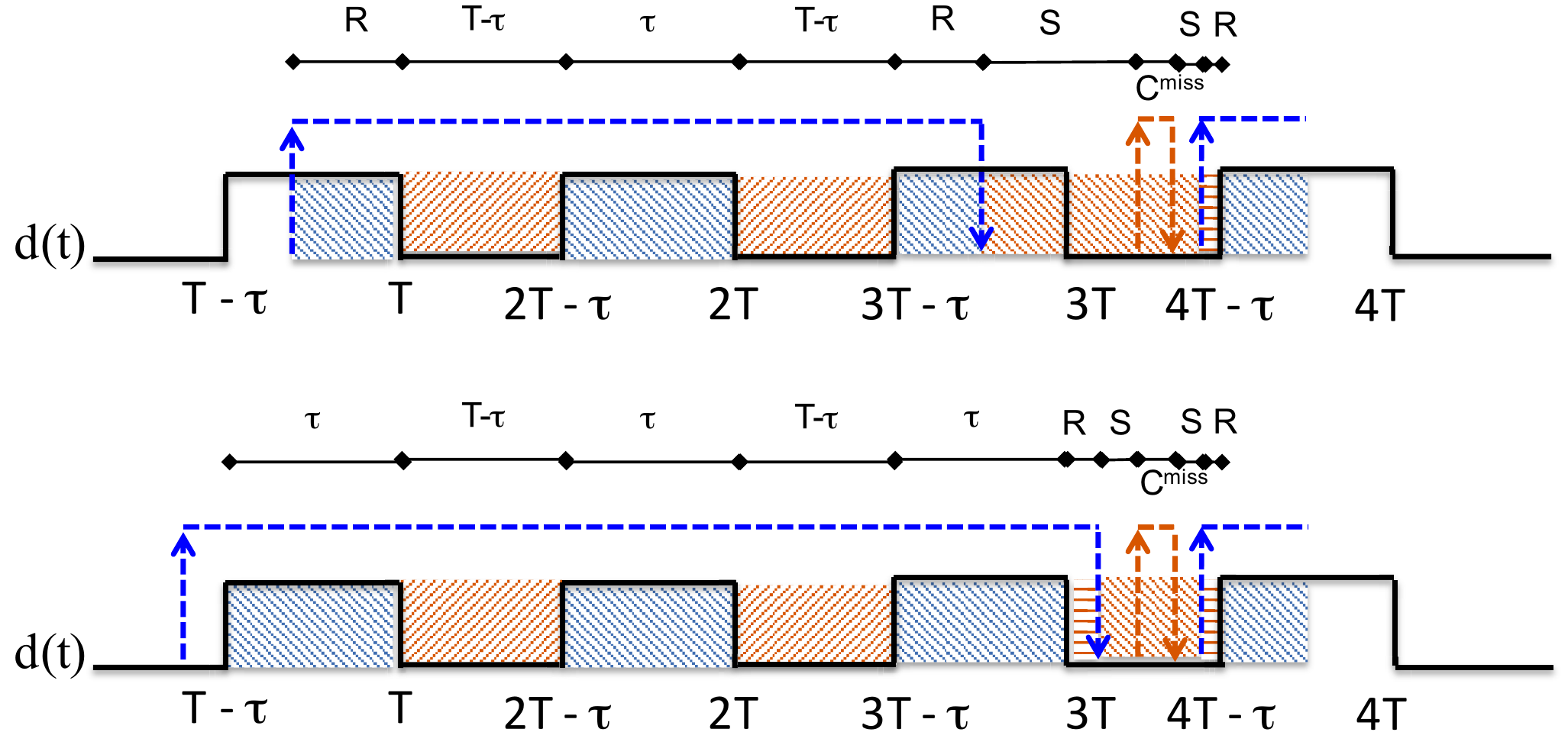}
\caption{Two examples of measured intercontact time when $H=3$}
\label{fig:detectedict_longct}
\end{center}
\end{figure}

Based on the above discussion, we know that measured intercontact times when contact duration is non-negligible can be either equal to $T-\tau$ (which is the contribution of pseudo-intercontact times) or to $R + \sum_{i=1}^{N} S_i + \sum_{i=2}^{N} C_i^{miss} + R$, thus the following theorem holds. \addCB{The weights of the two components of the mixture distribution are derived in}~\cite{biondi:what_you_lose:tr} according to the line of reasoning discussed above.

\begin{theorem}[Measured ICT]\label{theo:nonneg_measured_ict}
The measured intercontact time can be approximated as follows:
\begin{equation}\label{eq:theo:nonneg_measured_ict}
\tilde{S} = \left\{ \begin{array}{lr} 
R + \sum_{i=1}^{N} S_i +  \sum_{i=2}^{N} C_i^{miss}  + R & \frac{\sum_{h=1}^{\infty} P(H=h)}{\sum_{h=1}^{\infty} h P(H=h)}\\
T - \tau & \frac{\sum_{h=2}^{\infty} (h-1) P(H=h)}{\sum_{h=1}^{\infty} h P(H=h)}\\
\end{array}\right. \normalsize
\end{equation}
where $R$ is defined in Equation~\ref{eq:residuals_R} and the PDF of $C^{miss}$ is given by the following:
\begin{equation}\label{eq:pdf_c_miss}
\mathbb{P}\left(C^{miss} = c\right) =  \left\{ \begin{array}{lr}
 \frac{1}{T-\tau} f_C(c) \int_c^{T-\tau} \frac{1}{F_C(u)} \textrm{d} u & 0 < c < T - \tau \\
0 & \textrm{otherwise}
\end{array}\right.. 
\end{equation}
\end{theorem}


At this point, in order to obtain $\tilde{S}$ the only missing piece is the distribution of $N$ when contact duration is not negligible. Thus we derive it in the lemma below. The proof follows the same line of reasoning as the proof of Lemma~\ref{lemma:pmf_h}, and it can be found in~\cite{biondi:what_you_lose:tr}.

\begin{lemma}[PMF of $N$]\label{lemma:pdf_n_approx_nonneg}
When contact duration is non-negligible, the probability mass function of $N$ can be approximated by the following: 
\begin{equation}\label{eq:pdf_n_approx_nonneg}
\left\{ \begin{array}{lr}
\mathbb{P}\{N=1\} = \hat{g}, & k = 1\\
\mathbb{P}\{N=k\} = (1-\hat{g})(1-\hat{p})^{k-2}\hat{p}, & k\geq 2
\end{array}\right.
\end{equation}
where $\hat{g} = 1 - \sum_{n_1=2}^{\infty} \mathbb{P}\left( Z + S \in \mathcal{I}^{OFF}_{n_1}\right) \mathbb{P}\left( Z^{OFF} + C < T-\tau \right)$ and  \\ $\hat{p}=1- \sum_{n_2=2}^{\infty} \mathbb{P}\left( Z^{OFF} + S \in \mathcal{I}^{OFF}_{n_2}\right) \mathbb{P}\left( Z^{OFF} + C < T-\tau \right)$.
\end{lemma}
%

Going back to Theorem~\ref{theo:nonneg_measured_ict}, it is easy to see that, in most practical applications, $\tilde{S}$ is still dominated by the component $\sum_{i=1}^{N} S_i$. 
Thus, all the properties discussed in Section~\ref{sec:ict_negligible} (e.g., the conversion from hypo-exponential behaviour to a hyper-exponential one) will not change significantly. Please note, however, that the distribution of $N$ in the two cases is different, so the result obtained in the negligible contact duration case does not apply as is to the non-negligible contact duration case.

Finally, Theorem~\ref{theo:nonneg_measured_ict} can be easily modified to take into account the effect of nodes not entering the OFF state after a contact has been detected, which we have discussed in Section~\ref{sec:ict_non_negligible_measured_contact}. In this case, pseudo-intercontact times will not be present in $\tilde{S}$, since nodes do not enter the OFF state when a contact is ongoing. Thus, we obtain $\tilde{S} = R + \sum_{i=1}^{N} S_i +  \sum_{i=2}^{N} C_i^{miss}  + R$.

\subsection{Validation}
\label{sec:ict_non_negligible_validation}

We now validate the results we have obtained for the measured contact (Theorem~\ref{theo:detected_ct}) and measured intercontact time (Theorem~\ref{theo:nonneg_measured_ict}). \addCB{Specifically, as these are the most common distributions found in real traces, we consider real human mobility traces and extract pairs for which exponential or Pareto distributions fit intercontact or contact times. We use standard Maximum Likelihood estimation and goodness-of-fit techniques to this end.  Specifically, the method discussed in}~\cite{clauset2009power} for estimating both scale and shape for the Pareto distribution, and the Cram\'er-von Mises test with significance level $\phi = 0.01$. For statistical reliability, only pairs with more than $9$ samples are considered, as in~\cite{tournoux2011density}).

Several traces of real contacts between nodes are publicly available\footnote{E.g., at~\url{http://crawdad.cs.dartmouth.edu/}} and have been often used in the related literature (Table~\ref{tab:summary_traces} in~\cite{biondi:what_you_lose:tr} summarises the most popular ones). However, it is usually neglected the fact that all of them implement a form of duty cycling in the neighbour discovery process, owing to the technology-dependent scanning period (typically in the order of $100s$). Hence, what they track is actually the measured contacts and intercontact times, rather than the real ones. However, there are a few datasets that use quite a small duty cycling (in the order of a few seconds) and hence can realistically approximate the real contact and intercontact times in practice (i.e., assuming that both contact and intercontact times last for longer than a few seconds). These traces are PMTR~\cite{unimi-pmtr-20081201} and RollerNet~\cite{upmc-rollernet-20090202}, and they will be the focus of our analysis\footnote{In~\cite{biondi:what_you_lose:tr} we also provide a discussion on other traces (Infocom and Reality Mining), which are very popular in the literature but that were not suitable for our validation due to their long duty cycle.}. The PMTR trace has been obtained from the readings across 19 days (in November 2008) of $44$ Pocket Mobile Trace Recorders (PMTRs), custom devices built for contact detection and distributed to faculty members, PhD students, and technical staff at the University of Milan. Contacts are sampled every $1$ seconds. The RollerNet experiment was carried out to analyse the mobility of rollerbladers in Paris. The dataset was collected on August 20, 2006, and it is composed of two sessions of 80 minutes, interspersed with a break of 20 minutes. The $62$ Bluetooth sensors (iMotes) were distributed to organisers' friends, members of rollerblading associations and members of staff. Here contacts are sampled every $15$ seconds.

By applying the fitting technique described at the beginning of the section to each node pair in the PMTR and RollerNet traces, we obtained the results in Tables~\ref{tab:fitting_exp}-\ref{tab:fitting_pareto} in~\cite{biondi:what_you_lose:tr}. In the remaining of the section, we focus only on those pairs for which a given hypothesis (either exponential or Pareto) is not rejected, and we apply our theoretical framework to representative user pairs (i.e. we configure the model using the MLE parameters of the selected pair). In this analysis, we use the same duty cycle configuration that we used in Section~\ref{sec:ict_negligible_stilde_validation}, i.e., $\tau=20s$, $T=100s$. Assuming that the contact and intercontact times are either exponential (Sec.~\ref{sec:ict_non_negligible_validation_exp}) or Pareto (Sec.~\ref{sec:ict_non_negligible_validation_pareto}), we draw contact and intercontact times samples for either distribution, then we filter them according to the reference duty cycling process (hence, we simulate the effects of the duty cycling on the original contact process that we have extracted from the trace). Please note that the fitting results are used to configure the distribution from which contact and intercontact times are sampled for relevant pairs in the datasets.


\subsubsection{The exponential case}
\label{sec:ict_non_negligible_validation_exp}

We first consider the measured contact duration $\tilde{C}$, which we can approximate as discussed in Theorem~\ref{theo:detected_ct}. 
\addCB{The probabilities of observing each component of the mixture predicted in Theorem~\ref{theo:detected_ct} only depend on $\tau$, $T$ (that are constant in our case), and on the distribution of real contact times. Since we have fixed $\tau$ and $T$, the interplay between the different $\tilde{C}$ components is regulated only by~$C_i$, which we are assuming exponential with rate $\mu$ (a thorough discussion on this dependence is provided in}~\cite{biondi:what_you_lose:tr}). \addCB{Thus, in the following, for different $\mu$ values, we compare the predictions of Theorem~\ref{theo:detected_ct} for the measured contact time against simulation results.} Similarly to Section~\ref{sec:ict_negligible_pdf_n_validation}, we draw contact and intercontact times (30,000 samples each) from an exponential distribution and we filter them using our reference duty cycling process with $\tau=20s, T=100s$. We set the rate $\lambda$ of intercontact times equal to the mean rate in the corresponding trace (PMTR or RollerNet). Then, we set the rate of the exponential distribution of contact times equal to significant points of the $\mu$ distribution (corresponding to the minimum and maximum values, first and third quartiles, median and mean). Parameter $\mu$ varies a lot across pairs in the PMTR trace, spanning several orders of magnitude (from $10^{-5} s^{-1}$ to $10^{-1} s^{-1}$). For the RollerNet trace, the body of the distribution of $\mu$ is quite compact, and only the minimum and maximum values are more distant. Thus, in this case we omit the plot for the median and the first and third quartiles. Please note that we only consider those pairs for which the exponential hypothesis was not rejected by the Cram\'er-von Mises test. The results for the measured contact times are shown in Figure~\ref{fig:nonneg_exp_pmtr} for the PMTR trace and in Figure~\ref{fig:nonneg_exp_rollernet} for the RollerNet trace. In both cases, predictions are generally accurate. The largest discrepancies appear for values of $\mu$ around $10^{-2}$. \addCB{The figures also plot the distribution of the original contact time without duty cycling. As we have set $\tau=20$ and $T=80$, the maximum contact duration we observe with duty cycling is $20s$, while contact durations without duty cycling can be also much longer. Moreover, the model predicts a discontinuity in the density for contact around the value of $\tau$. This justifies the differences in the corresponding curves in the figures.} It is also interesting to note that when $\mu$ is large, the discontinuity in the CDF of the measured contact disappears, because contacts are so short that they are generally fully contained in ON intervals. Because of this, we expect that the original contact time distribution, $C$, is not significantly modified by the duty cycling. This is indeed the case, as we can observe\footnote{For large $\mu$ the main component of $\tilde{C}$ is $C^{res}$, which, in the exponential case, converges to $C$ for large $\mu$.} in Figure~\ref{fig:nonneg_exp_rollernet}.


\begin{figure}[t]
\begin{center}
\includegraphics[scale=0.6]{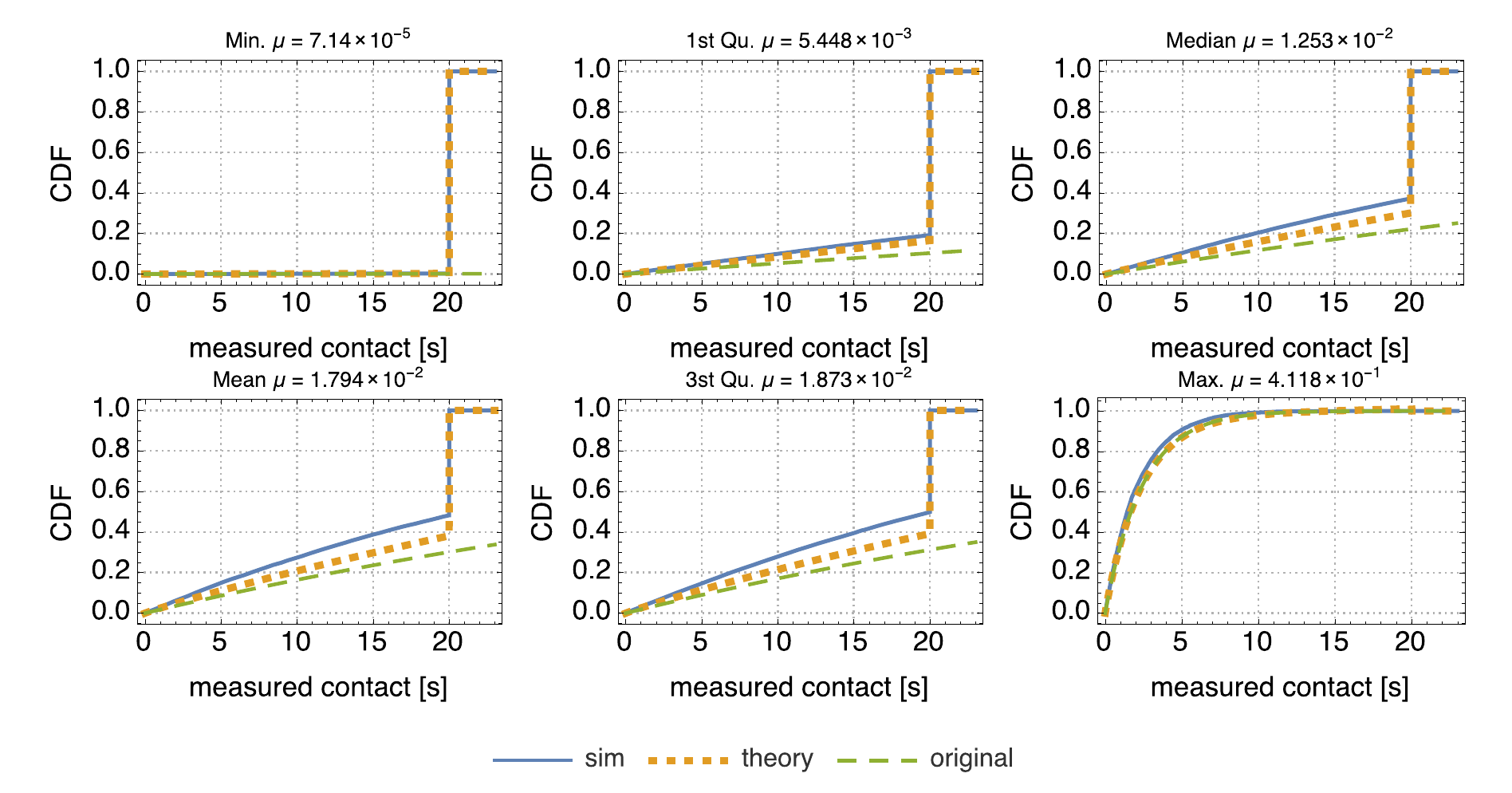}
\caption{Measured contact duration when real contact and intercontact times are exponential (PMTR).}
\label{fig:nonneg_exp_pmtr}
\end{center}
\end{figure}

\begin{figure}[h]
\begin{center}
\hspace{-5pt}\includegraphics[scale=0.64]{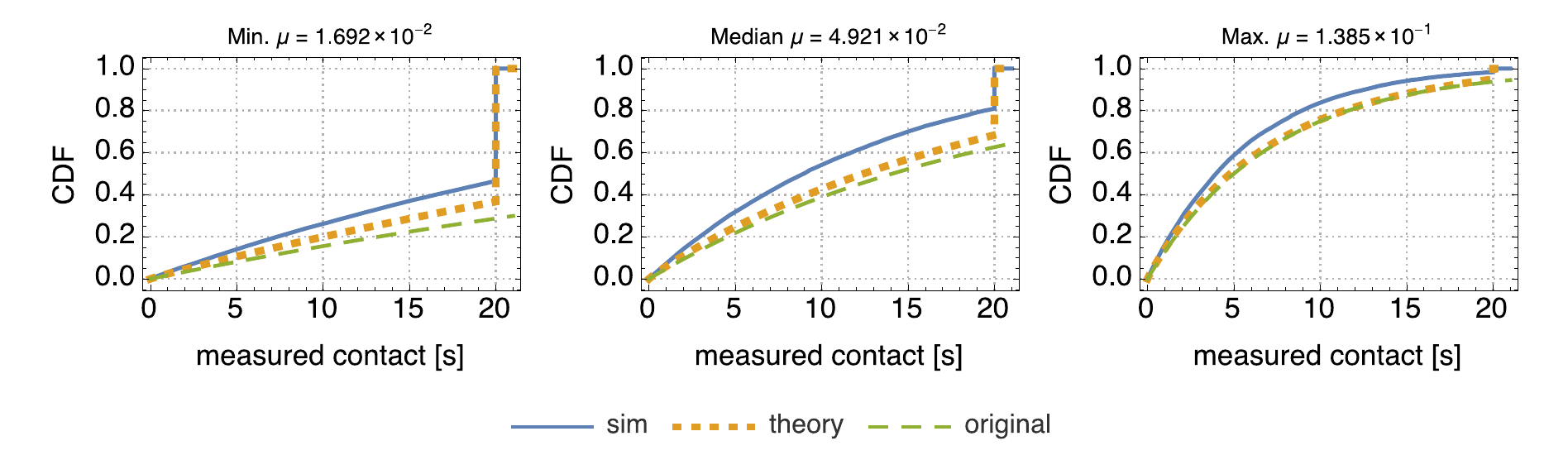}
\caption{Measured contact duration when real contact and intercontact times are exponential (RollerNet).}
\label{fig:nonneg_exp_rollernet}
\end{center}
\end{figure}

We now focus on the measured intercontact times $\tilde{S_i}$. Measured intercontact times have two components, one stochastic component which is the sum of several random variables (see Theorem~\ref{theo:nonneg_measured_ict}) and one constant component $T-\tau$ (which we have called pseudo-intercontact time). 
The probability of selecting either component only depends on $\mu$, $\tau$, and $T$ (because it is a function of $H$). Fixing, as usual, $\tau=20s$ and $T=100s$, we plot in Figure~\ref{fig:nonneg_exp_stilde_pseudo} the probability of observing pseudo-intercontact times in $\tilde{S_i}$. Pseudo-intercontact times start appearing for $\mu$ values smaller than $0.06$. In this range, long contacts are split into many shorter measured contacts, and the portions of real contacts not overlapping with an ON interval become pseudo-intercontact times.

In order to validate Theorem~\ref{theo:nonneg_measured_ict}, we start with the PMTR trace and we plot the CDF of the measured intercontact times keeping $\mu$ fixed (specifically, equal to the mean value) and varying $\lambda$. In all cases, theoretical predictions and simulations results are overlapping (Figure~\ref{fig:nonneg_exp_stilde_pmtr_fixed_mu}). We omit the plots for the RollerNet case, since the same considerations apply and theoretical predictions remain virtually indistinguishable from simulation results. We expected this good result, since $\lambda$ values are rather small (and, in all case but one, smaller than $1/T$), and therefore the slowly varying assumption of Theorem~\ref{theo:nonneg_measured_ict} is verified. In these plots we also observe the contribution of pseudo-intercontact times to the CDF, corresponding to its initial bump (pseudo-intercontact times are of length $T-\tau=80s$, hence they affect the very beginning of the distribution). 
This was also expected, as, with $\tau=20$ and $T=100$, the region of $\mu$ where the probability of pseudo intercontact times is not negligible is $\mu<0.06s^{-1}$, and the mean value of $\mu$ in PMTR is $0.02s^{-1}$.
Note also that with $\mu\sim0.02s^{-1}$, contact duration is not negligible, hence the predictions (green curve in Figure~\ref{fig:nonneg_exp_stilde_pmtr_fixed_mu}) that we would obtain using the results for the negligible contact case (i.e., using Lemma~\ref{lemma:stilde_neg_exp}) are not at all accurate.

\begin{figure}[t]
\begin{center}
\includegraphics[scale=0.6]{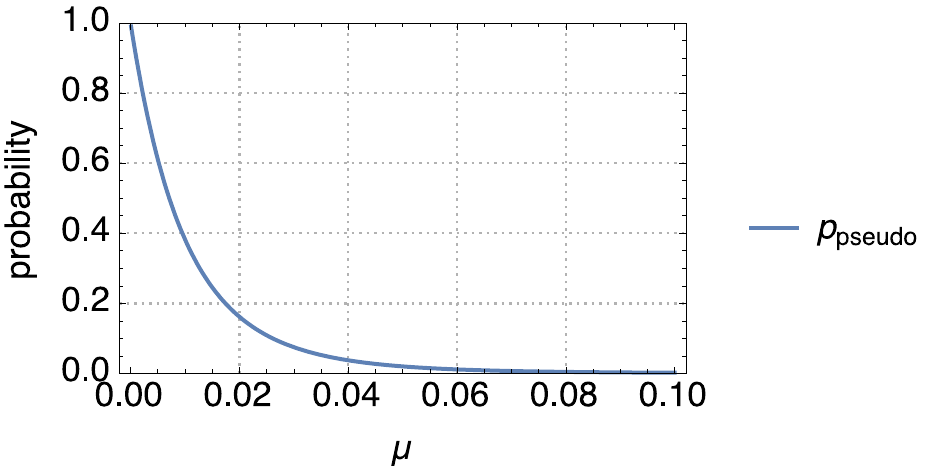}
\caption{Probability of observing pseudo-intercontact times in $\tilde{S}$.}
\label{fig:nonneg_exp_stilde_pseudo}
\end{center}
\end{figure}

\begin{figure}[t]
\begin{center}
\includegraphics[scale=0.6]{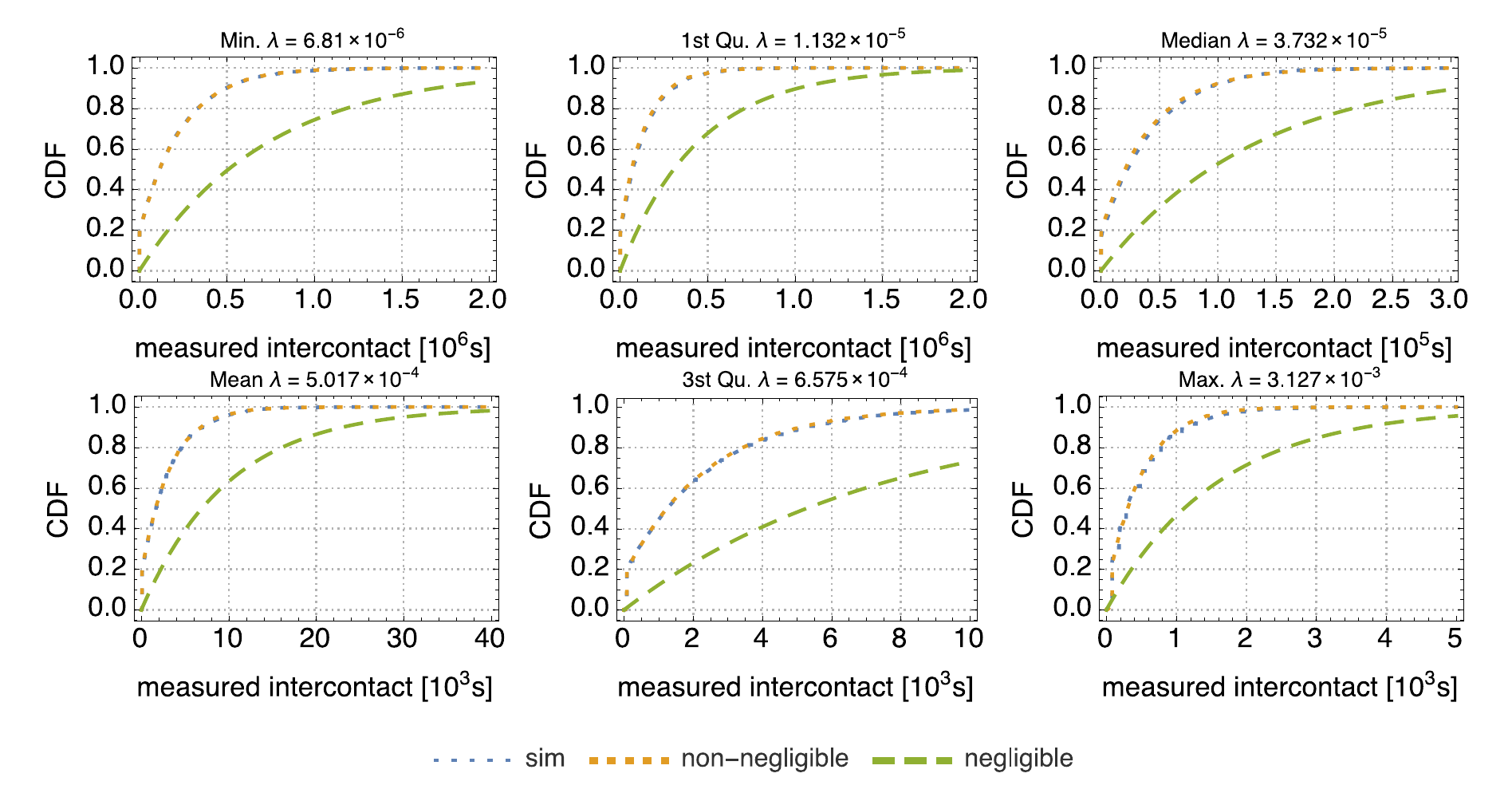}
\caption{Measured intercontact duration when real contact and intercontact times are exponential (PMTR).}
\label{fig:nonneg_exp_stilde_pmtr_fixed_mu}
\end{center}
\end{figure}

\subsubsection{The Pareto case}
\label{sec:ict_non_negligible_validation_pareto}
We now perform a similar analysis for those pairs for which the Pareto hypothesis for contact and intercontact times was not rejected by the Cram\'er-von Mises test (more than 97\% of pairs overall, see~\cite{biondi:what_you_lose:tr}). Here we only focus on measured contact and measured intercontact times. Further results, such as the analysis of $H$, can be found in~\cite{biondi:what_you_lose:tr}. In Figure~\ref{fig:nonneg_ct_pareto_pmtr} we plot the CDF of measured contact times obtained from simulations against the theoretical predictions of Theorem~\ref{theo:detected_ct}. Simulations are performed as described for the exponential case, except that here we sample from Pareto distributions. We start with the PMTR dataset. In the first set of plots in Figure~\ref{fig:nonneg_ct_pareto_pmtr_vara} we fix $b$ to the average value $b=245.40$ observed in the dataset and we vary $\alpha$. As expected, predictions are very accurate, owing to the fact that $b > T$ and thus the slowly varying assumption holds true. In Figure~\ref{fig:nonneg_ct_pareto_pmtr_varb} we vary $b$ fixing $\alpha$ to its average value ($\alpha =1.898$, which smaller than the threshold $\alpha=2$ discussed in~\cite{biondi:what_you_lose:tr} for having accurate predictions of $\tilde{C}$). \addCB{As expected, the only case when predictions are not very accurate are for $\alpha$ and $b$ small, when the slowly varying assumption does not hold.}
Similar considerations hold for the RollerNet dataset, for which we omit the plot.

\begin{figure}[t]
\begin{center}
\subfigure[Varying $\alpha$ ($b$ set to the mean $b=245.40$ in the dataset).\label{fig:nonneg_ct_pareto_pmtr_vara}]{\includegraphics[scale=0.6]{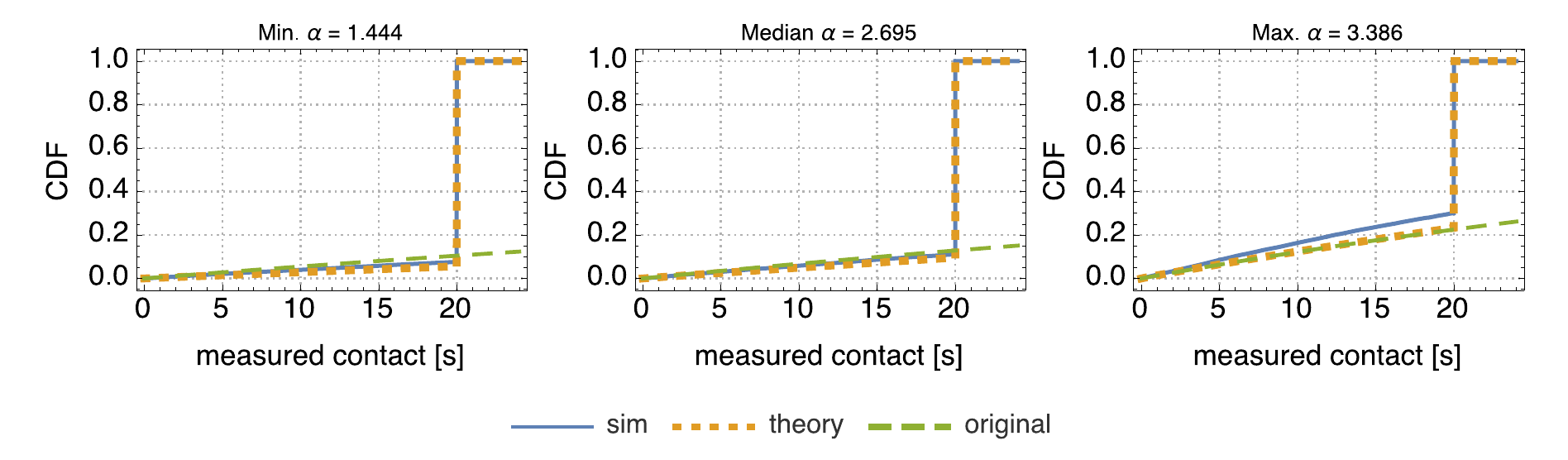}}
\subfigure[Varying $b$ ($\alpha$ set to the mean $\alpha=1.898$ in the dataset).\label{fig:nonneg_ct_pareto_pmtr_varb}]{\includegraphics[scale=0.6]{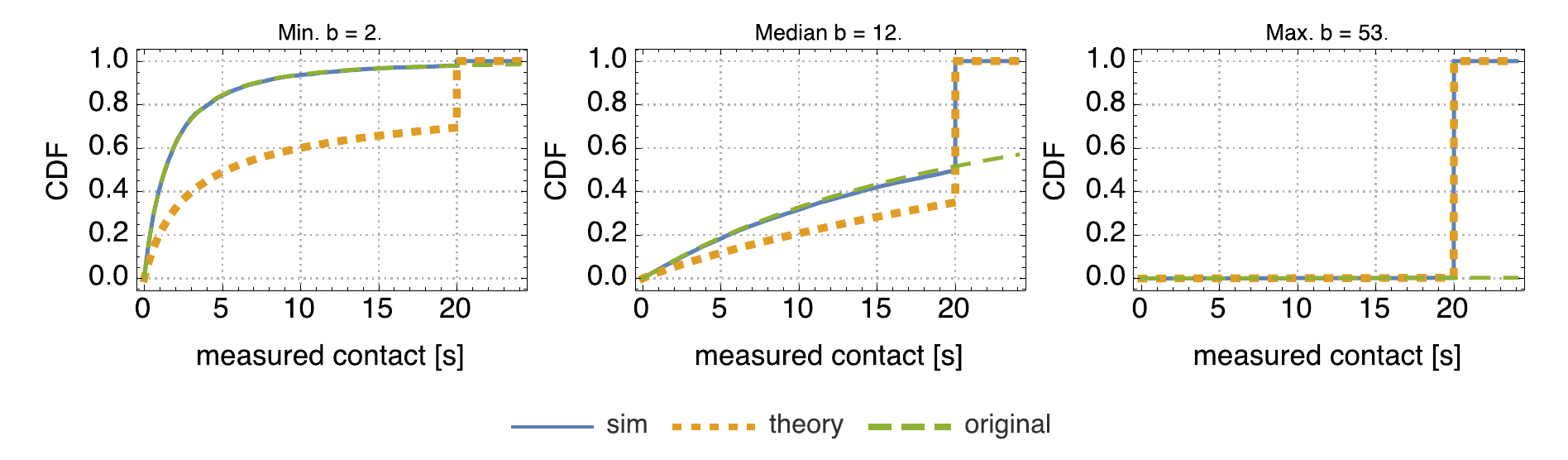}}
\caption{Measured contact duration when real contact and intercontact times are Pareto (PMTR).}
\label{fig:nonneg_ct_pareto_pmtr}
\end{center}
\end{figure}


\addCB{Finally, we study the behaviour of measured intercontact times when contact and intercontact times are Pareto.} 
In Figure~\ref{fig:nonneg_ict_par_pmtr} we plot the measured intercontact times (simulations vs theoretical predictions) fixing the Pareto parameters of the contact times to their average values 
and varying the parameters $b$ and $\alpha$ of intercontact times. We observe that in all cases the predictions are very accurate. Note how, in the PMTR+Pareto case, contacts tend to be long and pseudo-intercontact times are observed often, as the big jump at $T-\tau = 80s$ shows.
 

\begin{figure}[ht]
\begin{center}
\includegraphics[scale=0.6]{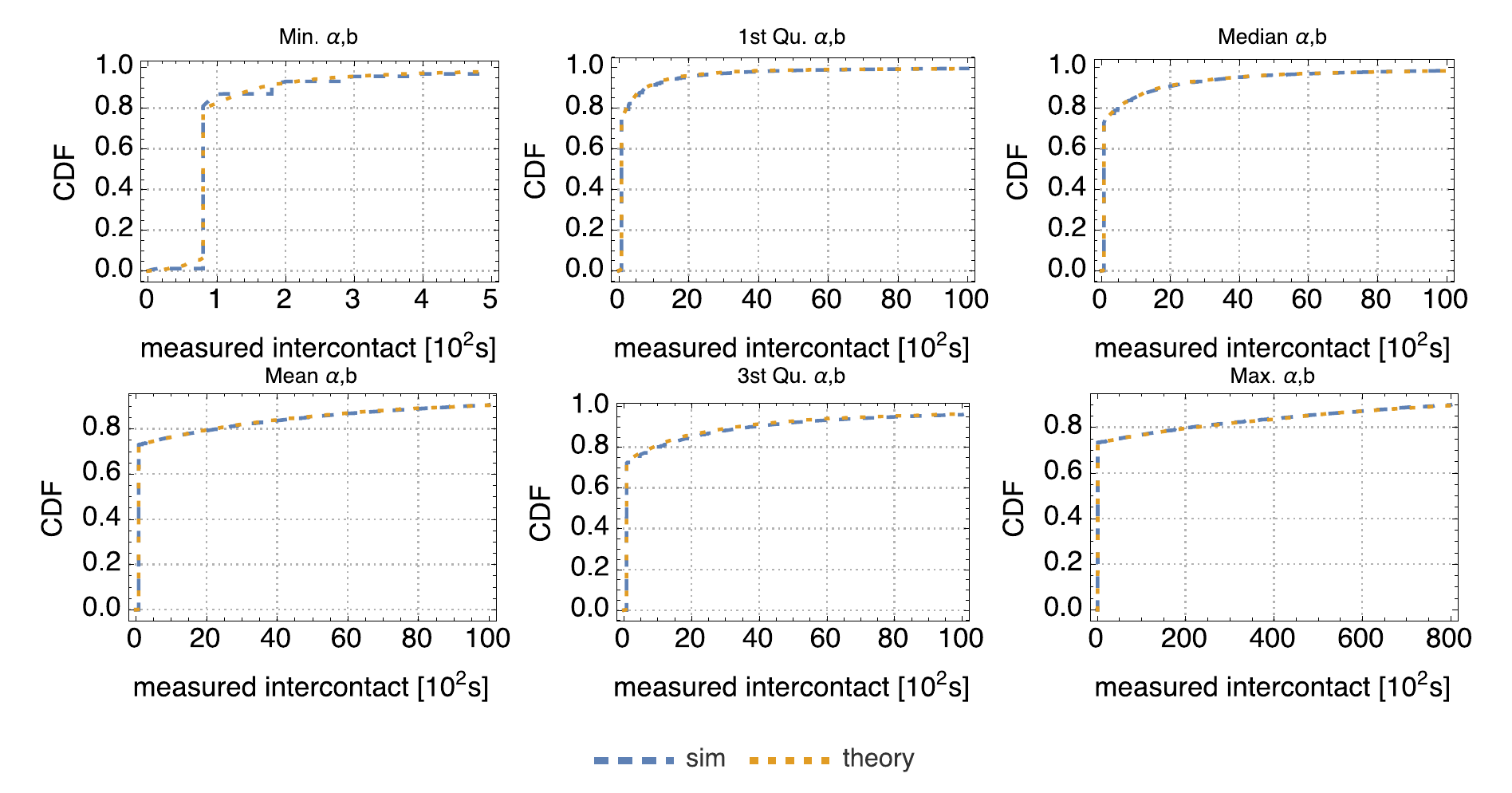}
\caption{Measured intercontact duration when real contact and intercontact times are Pareto (PMTR).}
\label{fig:nonneg_ict_par_pmtr}
\end{center}
\end{figure}


\section{From deterministic to stochastic duty cycling policies}
\label{sec:exp_dc}

In this section we show that the fixed duty cycling model studied so far in the paper provides a good approximation also for non deterministic duty cycling. Although the method can be generalised, for the sake of example, we assume that individual stochastic duty cycles alternate between OFF and ON intervals whose lengths are both exponentially distributed with rate $\alpha$ and $\beta$, respectively. 

The first step in the analysis is the derivation of the \emph{joint} duty cycling (which is in the ON state when both nodes are ON, in the OFF state otherwise). \addCB{$T_{ON}$ and $T_{OFF}$ denote the length of ON and OFF phases in the joint duty cycle.} 
\addCB{To this aim, we model the states of our system using a Continuous Time Markov Chain (CTMC). We obtain (the detailed derivation can be found in}~\cite{biondi:what_you_lose:tr}\addCB{) that $T_{ON}$ is exponentially distributed with rate~$2\beta$, and that $T_{OFF}$ has the following first and second moments:}
\begin{equation}\label{eq:first_and_second_moment_OFF_joint_stoc}
\mathbb{E}[T_{OFF}] = \frac{2\alpha + \beta}{2\alpha^2}, \quad \mathbb{E}[T_{OFF}^2] = \frac{10 \alpha ^3+11 \alpha ^2 \beta +6 \alpha  \beta ^2+\beta ^3}{2 \alpha ^4 (\alpha +\beta )}.
\end{equation}
We assume that both $\alpha$ and $\beta$ 
are strictly greater than zero\footnote{When equal to zero we have the two extreme cases of nodes either always ON or always OFF. In the first case, there is no need to study the effect of duty cycling, in the second case nodes are never able to detect each other.}. For $\beta$ approaching zero, the squared coefficient of variation of $T_{OFF}$ approaches 4. For $\alpha$ approaching zero, it approaches $1$. Hence, we can conclude that the duration of the OFF interval of the joint duty cycle ranges from a hyper-exponential behaviour to an exponential behaviour, depending on the values of $\alpha$ and $\beta$.

\addCB{The fixed joint duty cycling analysed in Sec.~\ref{sec:ict_negligible} can be considered as an approximation of the stochastic joint duty cycle studied in this section, setting the fixed duty cycling parameters $\tau$ and $T-\tau$ equal to the average ON and OFF durations of the joint duty cycle.}
%
Figure~\ref{fig:exp_vs_fixed} validates this statement by comparing simulation results obtained with exponential duty cycling against the prediction obtained assuming the duty cycling is deterministic. Specifically, we consider $\tau=20s$ and $T=100s$ for the fixed joint duty cycling (as in Section~\ref{sec:ict_negligible_stilde_validation}), thus we obtain for stochastic individual duty cycling $\beta=0.025$ and $\alpha=0.02$ (with $cv^2=1.96$ for the joint OFF intervals). We assume that real intercontact times are exponentially distributed with rates $\lambda \in \{0.001, 0.01, 0.1, 10 \}$ and that contact duration is negligible. We sample measured intercontact times using Monte Carlo simulations. \addCB{We test both stochastic and fixed duty cycling, and we compare the results obtained against the theoretical predictions of Theorem~\ref{theo:nonneg_measured_ict}, showing (Figure~\ref{fig:exp_vs_fixed}) that the model with fixed duty cycling well approximate also cases where duty cycling is stochastic, which is popular in the literature.}



\begin{figure}[t]
\begin{center}
\includegraphics[scale=0.48]{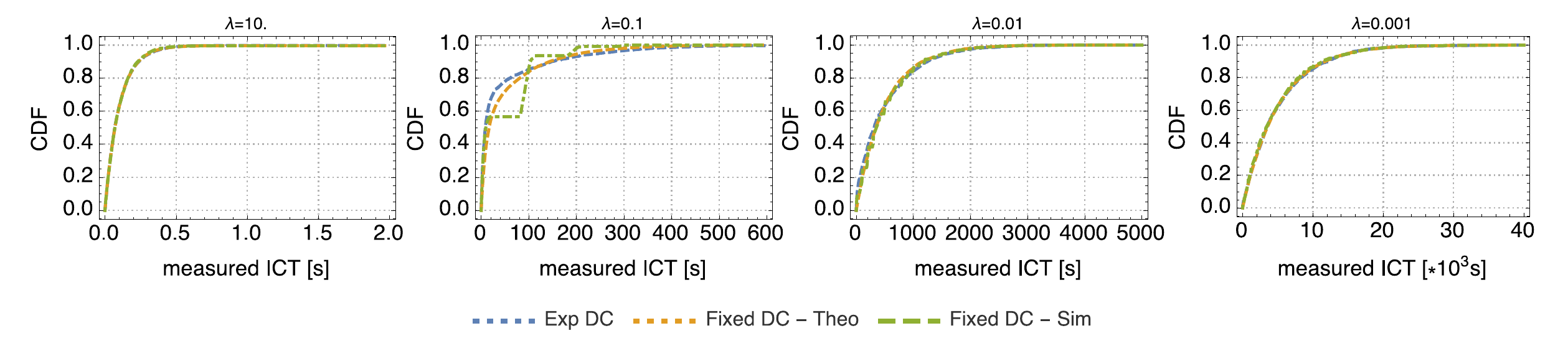}\vspace{-10pt}
\caption{Measured intercontact time with stochastic and deterministic duty cycles.}
\label{fig:exp_vs_fixed}\vspace{-10pt}
\end{center}
\end{figure}



\section{Related work}
\label{sec:relwork}

Ad hoc communications in opportunistic networks traditionally use either the WiFi or Bluetooth interfaces, which can consume a significant fraction of the smartphone's battery depending on the their current state (idle, scanning, or connected)~\cite{friedman2013power}. For this reason, duty cycling techniques have been introduced in order to save energy by putting devices into a low-power state whenever possible. With Bluetooth, this low-power state corresponds to the \emph{discoverable} state, which is entered by a device after a scanning phase (energy hungry) is concluded. Hence, duty cycling with Bluetooth implies striking the right balance between keeping as much as possible the devices in the discoverable state and not missing too many contacts. The situation is different with ad hoc WiFi,  which has no particularly energy-efficient state\footnote{For 802.11 cards used in \emph{infrastructure} mode energy consumption in the idle state has been drastically reduced by the introduction of the Power Saving Mode (PSM)~\cite{friedman2013power}. Unfortunately, PSM for ad hoc is typically not implemented in smartphones' 802.11 interfaces.}. In this case, the best power saving strategy is simply to switch off the network interface entirely. Recently, two innovative WiFi-based ad hoc communications modes, namely, WiFi Direct~and WLAN-Opp~\cite{trifunovic2011wifi}, have been proposed to address the problems of ad hoc communications in off-the-shelf smartphones. Unfortunately, their energy consumption is still quite high, in particular as far as neighbour discovery is concerned~\cite{trifunovic2014adaptive}. As discussed in~\cite{trifunovic2014adaptive}, reducing the scanning frequency remains the only viable power saving option for both WiFi Direct and WLAN-Opp. 

Abstracting the specific communication technology used and building upon the idea that neighbour discovery is an energy-expensive operation in general, the vast majority of papers dealing with power saving issues in opportunistic networks have focused on the contact probing phase. As seen above, reducing and optimising the probing frequency is equivalent to implementing a duty cycling policy in which nodes switch between low-power and high-power states, corresponding to OFF intervals (in which contacts are not detected) and ON intervals (in which contacts are detected), respectively.
Contact probing schemes can be classified into fixed, when the ON/OFF duration of the duty cycle is established at the beginning and never changed \cite{jun2005power,biondi14optimal}, or adaptive, when the frequency of probing is increased or decreased according to some policy \cite{choi2009adaptiveexponential}. Both fixed and adaptive strategies can be context-oblivious \cite{jun2005power}, if they do not exploit information on user past behaviour or position, or context-aware \cite{drula2007adaptive,wang2009opportunistic,bracciale2016sleepy} otherwise. 
Differently from the above contributions, in this work we do not aim at deriving an optimised power saving strategy for DTN. Instead, our goal is to understand, \emph{given a duty cycling strategy}, how the measured contact process between nodes is changed by this strategy. Below, we briefly contrast the most relevant related literature against our contribution.

The model we introduce in this paper is more general than the one discussed in \cite{zhou2016energy-efficiency,zhou2013energy} as it is not bound to the RWP model but it can be applied to any distribution for intercontact times. If the intercontact times distributions are Pareto or exponential -- also in heterogeneous cases where the parameters change across pairs of nodes -- our model can be solved with closed form expressions (note that  Pareto and exponential are the two most popular assumptions for contact and intercontact times in the related literature). Otherwise, for any other intercontact time distribution,  numerical solutions can be found. In addition, \cite{zhou2016energy-efficiency,zhou2013energy} only consider a fixed probing every $T$ seconds as their duty cycling strategy. Instead, we generalise the duty cycling process, by considering it composed of two phases (the ON and OFF phases) and also covering the non-fixed duration case.
\cite{qin2011contact} evaluate only how link duration (or contact duration, in our terminology) is affected by the contact probing interval. Instead, we investigate the effect of duty cycling (which, as already discussed, can be easily translated into a contact probing problem) \emph{both} on measured contact duration (i.e., link duration) and measured intercontact time, acknowledging that both components have a huge impact on opportunistic communications (on network capacity and message delay, respectively). Also, as already discussed, despite its simplicity, our duty cycling function with ON/OFF states allows for more flexibility than the simple scanning every $T$ seconds performed in~\cite{qin2011contact}. 
The effective link duration (equivalent to our measured contact duration) is also derived in~\cite{kouyoumdjieva2016impact}, assuming that nodes wake up every $T$ seconds and remain active for a configurable random amount of time. In this work, the authors implicitly discard the correlations between consecutive contacts (for which we have provided a thorough discussion in~\ref{sec:ict_negligible_preliminaries}), do not provide closed-form results for particularly relevant case, and do not investigate nor validate in detail the model (because the focus of the work is more on the energy-goodput trade-offs than on the effects of duty cycling on the contact process). In addition, measured intercontact times are not studied in~\cite{kouyoumdjieva2016impact}, despite their importance. 
Another set of works that share similarities with our proposal are~\cite{zhou2013exploiting-contact,zhou2012energy-saving}. Analogously to~\cite{kouyoumdjieva2016impact}, their focus is more on striking the right balance between energy consumption and forwarding performance rather than on the complete characterisation of the measured contact process. \cite{zhou2013exploiting-contact,zhou2012energy-saving}~assume the same kind of duty cycling process with ON/OFF periods that we study in this paper, while for the contact process they assume Pareto contact duration and exponential intercontact times. Under these assumptions, \cite{zhou2013exploiting-contact} derive that the exponential intercontact times are altered by duty cycling in such a way that their rate is scaled by a factor that they call contact probability. This result is analogous to our result in Lemma~\ref{lemma:stilde_neg_exp}. However, this result only holds for contact duration negligible with respect to $T$, and the model in \cite{zhou2013exploiting-contact,zhou2012energy-saving} is not able to address what happens when contact duration is instead not negligible. Our model is able to provide a complete characterisation also for this case. In addition, \cite{zhou2013exploiting-contact,zhou2012energy-saving}~do not provide an expression for the distribution of the measured contact and intercontact times under generic distributions of  contact and intercontact times. Instead, we address this case and provide a general technique, based on the slowly varying approximation, for handling distributions that are not memoryless.
Differently from the works discussed so far, in which the wake up schedules of nodes are fixed, \cite{zhou2014adaptive-working} and \cite{gao2013wakeup-scheduling} propose techniques to adaptively schedule the wake-up and sleep states of nodes. To this aim, and differently from our work,  \cite{zhou2014adaptive-working} do not characterise the contact process using a probability distribution but instead rely directly on the history of past encounters. Therefore, our model is more general, as it can represent in a mathematical form the contact process, and the effect on it of duty cycling. \cite{gao2013wakeup-scheduling}~do not provide a complete characterisation of the impact of duty cycling on the measured contact process either, but focus on probabilistically predicting the next contact. This probability is then used to design their adaptive wake-up schedule. \cite{gao2013wakeup-scheduling}~assume that intercontact times are exponential, and also neglect contact duration. Instead, we additionally consider the Pareto case for intercontact times and we include the effect of contact duration in the model.

Based on the above review, we can conclude that our contribution represents the first comprehensive analysis of how the measured contact process is altered by power saving techniques, both in terms of the effects on the measured contact duration and on the measured intercontact times.
This work is an extension of our previous work in~\cite{biondi2014duty}, where we had focused on the negligible contact case with exponential real intercontact times only, and we had studied how their distribution was affected by the duty cycling policy. In~\cite{biondi2014duty} we have used a complex model, which was not suitable to be solved with distributions different from the exponential. \deleted{Specifically, closed-form solutions could not be found when intercontact times were not exponential, and also numerical solution took a lot of time to be obtained.} The main outcome of the model in~\cite{biondi2014duty} is what we have here summarised in Lemma~\ref{lemma:stilde_neg_exp}.

\color{black}
\vspace{-10pt}
\section{Conclusions}
\label{sec:conclusion}

Power saving mechanisms reduce the forwarding opportunities and the capacity of an opportunistic networks, but this effect has not been yet quantified in a general setting in the related literature. To fill this gap, in this work we have investigated the effects of deterministic duty cycling on contact and intercontact times in opportunistic networks. Specifically, we have proposed two models for characterising the measured contact process (i.e., the contact process after duty cycling has been factored in) between pairs of nodes. These models have been extensively validated, and have been shown to provide very good approximations even when the assumptions under which they have been derived do not hold exactly. 

The first model can be used when the contact duration for the pair of nodes is negligible with respect to the length of the ON and OFF intervals of the duty cycle. With this model we can derive the first two moments of the measured intercontact times for any distribution of  original intercontact times. Exploiting this model, we have discovered that if the original intercontact times are exponential, then the measured intercontact times are also exponential but with a different rate. If original intercontact times are Pareto, measured intercontact times remain Pareto with the same exponent in the tail but, overall, they do not feature a well-known distribution. More in general, we have shown that the measured intercontact times can flip their ``behaviour'' depending on the duty cycle value and on the distribution of the original intercontact times. Specifically, hyper-exponential measured intercontact times can appear even when the original intercontact times are hypo-exponential, and vice versa.

The second model, which is more complex but also more realistic, should be used when contact duration is not negligible. With this second model, we are able to derive the distribution of the measured contact duration, considering both the case in which nodes keep their scheduled duty cycle upon a new encounter and the case in which they do not. In the first case, a measured contact cannot last longer than an ON interval. Since contact duration determines the amount of data that can be transferred, the capacity of the opportunistic network can be significantly affected by duty cycling. Vice versa, in the second case, only a small portion of the contact is missed, hence the capacity can be preserved. We have also derived the measured intercontact times, highlighting the fact that they have two components: one conceptually very similar to the measured intercontact time with negligible contact duration, one very different. We called the second component pseudo-intercontact time, as it is due to long contacts that are split into many shorter contact and intercontact times by the duty cycle.

Finally, we have generalised our results, showing that a deterministic duty cycle can be assumed to be a good approximation of a stochastic duty cycle with the same average duration for the ON/OFF intervals. Building upon this finding, the two models presented in the paper can be used to derive the measured contact and intercontact times under any distribution for the contact process and under general duty cycling strategies (deterministic/stochastic, synchronous/asynchronous). 

\vspace{-5pt}
%

%

\bibliographystyle{abbrv}
\bibliography{tompecs.bib}

\end{document}